\documentclass[11pt,a4paper]{article}

\usepackage{amsmath,amssymb,amsfonts,amsthm}
\newtheorem{theorem}{Theorem}

\newtheorem{remark}{Remark}
\newtheorem{proposition}{Proposition} 

\usepackage[margin=1.25in]{geometry}
\usepackage[authoryear,round]{natbib}
\usepackage{graphicx} 
\usepackage{epstopdf}
\usepackage[dvipsnames]{xcolor}
\usepackage[colorlinks=true,linkcolor=red,citecolor=red,urlcolor=blue]{hyperref}
\usepackage[title]{appendix}
\usepackage{setspace}
\usepackage{booktabs}
\usepackage[inline]{enumitem}
\usepackage{multirow}
\usepackage{subcaption}
\theoremstyle{definition}

\title{\textbf{Robust Interpolated Quantile Estimators: Asymptotic Theory and Efficiency}}
\usepackage{authblk}

\author[1]{Sa\"{i}d Maanan\thanks{Corresponding author: \texttt{maanan.said@gmail.com}}}
\author[2]{Azzouz Dermoune}
\author[1]{Ahmed El Ghini}
\author[3]{Daoud Ounaissi}

\affil[1]{LEAM, Mohammed V University in Rabat,  Morocco}
\affil[2]{Laboratoire Paul Painlevé, USTL–CNRS–UMR 8524, Lille, France}
\affil[3]{ESA, Angers, France}

\date{}
\newcommand{\keywords}[1]{\par\noindent\textbf{Keywords:} #1}

\begin{document}
\maketitle
\begin{abstract}
This paper introduces a unified family of interpolated quantile estimators obtained by augmenting the check loss with quadratic, Huber, or Tukey's bisquare regularization. The estimators are indexed by the quantile level $\tau$ and an interpolation parameter $h$. They reduce to the classical empirical quantile when $h=0$, while increasing $h$ continuously shifts the effective probability level toward the center of the distribution.
A complete asymptotic theory is developed. For the quadratic interpolation, the effective quantile level is characterized by an interpolation equation yielding a closed-form parametrization of neighboring quantiles. Asymptotic normality is established for all three interpolated estimators via M-estimation, and a decomposition of the asymptotic variance explains how efficiency depends on the underlying distribution.
Numerical experiments show that the quadratic interpolated estimator can reduce asymptotic variance by up to 36\% for light-tailed distributions and up to 57\% for heavy-tailed or asymmetric distributions for suitable interpolation strength. The framework is extended to linear quantile regression, where Monte Carlo experiments show that Huber interpolation is beneficial only in a narrow neighborhood of the median, while ordinary quantile regression remains preferable elsewhere. An application to daily log-returns illustrates the practical relevance of the proposed methodology for tail estimation under heavy tails and asymmetry.
\end{abstract}
\par\noindent\textbf{MSC2020 subject classifications:} 62G30, 62G20, 62G35, 62F35, 62J05.

\clearpage
\section{Introduction}

Quantile estimation is a central problem in statistics, with applications ranging from risk measurement in finance to robust inference in econometrics and survival analysis. The classical empirical quantile estimator, defined as an order statistic, is simple and distribution-free, but its discontinuous influence function can induce substantial finite-sample variability, particularly in the tails of the distribution or under heavy-tailed noise. This has motivated a long line of work on smoothed, regularized, and robust alternatives.

The quantile regression estimator of \citet{Koenker1978} extends scalar quantile estimation to the linear model and has become the standard tool for conditional quantile inference. Under the asymmetric Laplace model, it coincides with the maximum likelihood estimator and attains the Cramér--Rao lower bound asymptotically \citep{Yu2001, Komunjer2005}. Despite this optimality, the non-differentiability of the check loss can cause instability in finite samples, a phenomenon well documented in the robust statistics literature \citep{Hampel2011}. Smoothed versions of the check loss have been studied as a remedy \citep{Fernandes2019, Kaplan2016}, as have kernel-based and bandwidth-regularized quantile estimators \citep{Chaudhuri1991, Yu1998}.

A parallel line of research concerns M-estimators that interpolate between different loss functions. The Huber estimator \citep{Huber1964} provides a canonical example, bridging least squares and least absolute deviations through a quadratic-linear loss. In the quantile regression context, similar ideas have appeared in proposals that blend the check loss with quadratic or other smooth penalties \citep{Belloni2011}, often motivated by computational, robustness, or shrinkage considerations. In contrast, our focus is on the induced deformation of the effective quantile level and its probabilistic interpretation.

The present paper takes a different and more unified approach. We introduce three families of interpolated quantile estimators obtained by adding a robust regularization term to the check loss: a quadratic penalty, a Huber penalty, and Tukey's bisquare penalty \citep{Beaton1974}. Rather than treating the regularization parameter purely as a tuning device, we show that it induces a well-defined and monotone deformation of the effective quantile level, so that each interpolated estimator may be interpreted as a generalized empirical quantile at an implicitly defined probability level. This probability-space perspective unifies the three constructions and connects them to the classical theory of empirical quantile estimation.

Our main contributions are the following. First, we establish an asymptotic theory for the quadratic interpolation estimator, including a characterization of the effective quantile level through an interpolation equation, asymptotic normality via M-estimation theory, an explicit decomposition of the asymptotic variance identifying sufficient conditions under which the interpolated estimator is asymptotically more efficient than the empirical quantile estimator. Second, we derive the asymptotic distribution of the Huber interpolation estimator and characterize its probability-space deformation, highlighting the role of bounded influence in limiting the shift of the effective quantile level. Third, we study the bisquare interpolation estimator, establishing existence of solutions and asymptotic normality, and discussing its redescending influence function as a mechanism for discarding the contribution of extreme residuals. 
Finally, we extend the framework to linear quantile regression and investigate its finite-sample behavior through a comprehensive Monte Carlo study across the full range of quantile levels. The experiments show that the finite-sample benefit of Huber interpolation is confined to a narrow neighborhood of the median quantile, while the classical quantile regression estimator remains preferable over most of the quantile range. The methodology is further illustrated through an application to CAC-40 daily log-returns, a setting exhibiting the heavy tails and asymmetry that motivate the proposed framework.

The remainder of the paper is organized as follows. Section~\ref{subsec:inter_estim} introduces the generalized empirical quantile framework and defines the three interpolated estimators. Section~\ref{sec:one} develops the complete theory for quadratic interpolation. Sections~\ref{sec:two} and~\ref{sec:three} treat the Huber and bisquare interpolations respectively. Section~\ref{sec:reg} extends the framework to linear regression. Section~\ref{sec:cac40} provides real-data illustrations using CAC-40 returns.

\section{Generalized empirical quantiles}
\label{subsec:inter_estim}

We begin by recalling the classical empirical quantile and then extend the concept to a broader class of estimators.
Let
\(y_{(1)} \le \cdots \le y_{(n)}\)
denote the order statistics and
let
\[
F_n(x)
=
\frac1n
\sum_{i=1}^n
\mathbf 1_{\{y_i\le x\}}
\]
be the empirical distribution function.
The classical empirical quantile estimator of order $\tau$ is given by the order statistic $y_{(k)}$, where the integer $k$ satisfies
\(
\frac{k-1}{n}
<
\tau
\le
\frac{k}{n}.
\)
Equivalently,
\(
F_n(y_{(k)})
=
\frac{k}{n}.
\)
Under standard regularity assumptions,
\(
F_n(y_{(k)})
\to
\tau
\) and
\(
y_{(k)}
\to
q(\tau),
\)
where $q(\tau)$ denotes the theoretical quantile of order $\tau$.
This characterization suggests a broader interpretation of quantile estimation.
More generally, any sequence of estimators \(q_n^*\) satisfying
\(q_n^*\to q(p)\) and \(F_n(q_n^*)\to p\)
may be interpreted as estimating the quantile of order \(p\), 
even when the level \(p\) is not fixed a priori through an order statistic.

\subsection{Generalized empirical quantiles and interpolated estimators}

The objective of this paper is to investigate three families of interpolated quantile estimators obtained by combining the check loss with a robust regularization term.
The quantile check loss is given by \(\rho_\tau(u)=u\bigl(\tau-\mathbf 1_{\{u<0\}}\bigr),\)
while \(\rho_H\) and \(\rho_B\) denote the Huber and Tukey bisquare loss functions, respectively, whose explicit expressions are introduced in the corresponding sections below.
For a sample $(y_1,\ldots,y_n)$, the first estimator is the quadratic interpolation estimator
\[
\widehat q_Q(\tau,h)
=
\arg\min_{q\in\mathbb R}
\left\{
\sum_{i=1}^n
\rho_\tau(y_i-q)
+
\frac{h}{2}
\sum_{i=1}^n
(y_i-q)^2
\right\}.
\]
The second estimator is obtained by replacing the quadratic penalty with the Huber loss,
\[
\widehat q_H(\tau,h)
=
\arg\min_{q\in\mathbb R}
\left\{
\sum_{i=1}^n
\rho_\tau(y_i-q)
+
h
\sum_{i=1}^n
\rho_H(y_i-q)
\right\}.
\]
The third estimator uses Tukey's bisquare loss \citep{Beaton1974},
\[
\widehat q_B(\tau,h)
=
\arg\min_{q\in\mathbb R}
\left\{
\sum_{i=1}^n
\rho_\tau(y_i-q)
+
h
\sum_{i=1}^n
\rho_B(y_i-q)
\right\}.
\]
Each estimator depends on a quantile level $\tau\in(0,1)$ and an interpolation parameter $h\ge0$.
When $h=0$, all three estimators reduce to the classical empirical quantile.
As $h$ increases, the regularization term progressively modifies the target quantile, yielding a class of interpolated estimators whose theoretical properties require appropriate existence, uniqueness, and regularity conditions. Indeed, although these estimators are obtained through penalized optimization problems rather than direct order statistics, they remain associated with empirical probability levels through the empirical distribution function.

For the interpolated estimators, the empirical probability level is no longer fixed a priori.
Instead, each regularization parameter $h$ induces an effective empirical quantile order through the empirical distribution function. For example, for the quadratic interpolation, we define
\[
\hat p_Q(\tau,h)
=
F_n(\hat q_Q(\tau,h)).
\]
Similarly, the Huber and bisquare interpolations induce the empirical levels
\[
\hat p_H(\tau,h)
=
F_n(\hat q_H(\tau,h)),
\]
and
\[
\hat p_B(\tau,h)
=
F_n(\hat q_B(\tau,h)).
\]
When $h=0$, all three interpolated estimators coincide with the classical empirical quantile, so that
\[
\hat p_Q(\tau,0)
=
\hat p_H(\tau,0)
=
\hat p_B(\tau,0)
=
F_n(y_{(k)})
=
\frac{k}{n},
\]
with
\[
\frac{k-1}{n}<\tau\le\frac{k}{n}.
\]
In the quadratic interpolation case, as \(h\to\infty\),
\(
\hat q_Q(\tau,h)\to\bar y,
\)
so that
\(
\hat p_Q(\tau,h)\to F_n(\bar y).
\)
Similarly, for the Huber interpolation,
\(
\hat q_H(\tau,h)
\to
\hat\theta_H
\)
as
\(
h\to\infty,
\)
where $\hat\theta_H$ is the Huber location estimator, defined as the minimizer of
\(
\sum_{i=1}^n
\rho_H(y_i-\theta).
\)
Hence, as \(h\to\infty\),
\(
\hat p_H(\tau,h)
\to
F_n(\hat\theta_H).
\)
Consequently, under appropriate existence, uniqueness, and regularity conditions, the interpolated estimators considered in this paper may be viewed as defining a new family of generalized empirical quantile estimators combining geometric regularization, robustness, and quantile-type asymptotic behavior.
The remainder of the paper studies these three interpolated quantile estimators separately.
We begin with the quadratic interpolation estimator, for which a complete asymptotic theory is established.
We then investigate the Huber interpolation estimator and finally the Tukey bisquare interpolation estimator.

\section{Quadratic interpolation}
\label{sec:one}
\subsection{Interpolation equation and parametrization}   
\label{subsec:oneone}

Let $q(\tau)=F^{-1}(\tau)$
denote the theoretical quantile of order
$\tau\in(0,1),$
where \(F\) is the distribution function of a real-valued random variable \(Y\).
The classical quantile identity is
$F(q(\tau))=\tau.$
Let
$m=\mathbb E(Y)$
denote the population mean.
Given a reference probability level
$\tau\in(0,1),$
we seek to parametrize neighboring probability levels through a deformation parameter linking quantiles to the central location \(m\).
More precisely, for every probability level $\tau'$ located between $\tau$ and $F(m),$ we define the interpolation parameter
\[
h(\tau,\tau')
=
\frac{\tau-\tau'}
{q(\tau')-m},
\]
provided
$q(\tau')\neq m.$
The parameter
$h(\tau,\tau')$
measures the ratio between the probability displacement
$\tau-\tau'$
and the corresponding geometric displacement
$q(\tau')-m$
relative to the mean.
Using the quantile identity
$F(q(\tau'))=\tau',$
we obtain
\[
F(q(\tau'))
+
h(\tau,\tau')
\bigl(q(\tau')-m\bigr)
=
\tau.
\]
This naturally leads to the interpolation equation
\[
F(q)+h(q-m)=\tau.
\]
The following proposition shows that the interpolation parameter
$h(\tau,\tau')$
provides an exact parametrization of probability levels through the interpolation equation.

\begin{proposition}
	\label{prop:parametrization}
	
	Assume that the distribution function \(F\) is strictly increasing on the interval
	$[\min\{q(\tau'),m\},\max\{q(\tau'),m\}].$
	Let
	$\tau\in(0,1),$
	and let
	$\tau'$
	belong to the interval between
	$\tau$
	and
	$F(m).$
	Define
	\[
	h(\tau,\tau')
	=
	\frac{\tau-\tau'}
	{q(\tau')-m}.
	\]
	Then
	$q(\tau')$
	is the unique solution of
	$F(q)+h(\tau,\tau')(q-m)=\tau.$
	Consequently, if
	$q(\tau,h)$
	denotes the solution of the interpolation equation
	$F(q)+h(q-m)=\tau,$
	then
	$q(\tau,h(\tau,\tau'))=q(\tau')$
	and
	$F(q(\tau,h(\tau,\tau')))=\tau'.$
\end{proposition}

\begin{proof}
	
	Since
	$F(q(\tau'))=\tau',$
	one has
	\[
	F(q(\tau'))
	+
	h(\tau,\tau')
	\bigl(q(\tau')-m\bigr)
	=
	\tau'
	+
	\frac{\tau-\tau'}
	{q(\tau')-m}
	\bigl(q(\tau')-m\bigr).
	\]
	Hence
	\[
	F(q(\tau'))
	+
	h(\tau,\tau')
	\bigl(q(\tau')-m\bigr)
	=
	\tau.
	\]
	Therefore,
	$q(\tau')$
	satisfies the interpolation equation
	$F(q)+h(\tau,\tau')(q-m)=\tau.$
	Since \(F\) is strictly increasing, the function
	$q\mapsto F(q)+h(q-m)$
	is also strictly increasing.
	Hence the solution is unique, implying
	$q(\tau,h(\tau,\tau'))=q(\tau').$
	Applying \(F\) to both sides yields
	$F(q(\tau,h(\tau,\tau')))=\tau'.$
\end{proof}

\subsection{Variational formulation}                      

The probability-space parametrization introduced in subsection~\ref{subsec:oneone} admits a natural variational interpretation based on penalized quantile regression.
For a fixed quantile level
$\tau\in(0,1),$
define the quantile check loss
$\rho_\tau(u)=u\bigl(\tau-\mathbf 1_{\{u<0\}}\bigr).$
For every interpolation parameter
$h\ge0,$
consider the population optimization problem
\[
q_Q(\tau,h)
=
\arg\min_{q\in\mathbb R}
\left\{
\mathbb E[\rho_\tau(Y-q)]
+
\frac h2
\mathbb E[(Y-q)^2]
\right\}.
\]
The following proposition shows that the minimizer of this variational problem coincides with the solution of the interpolation equation introduced previously.

\begin{proposition}
	\label{prop:variational-characterization}
	
	Assume that
	$\mathbb E(Y^2)<\infty,$
	and that the distribution function \(F\) is continuous and strictly increasing.
	Then the minimizer
	$q_Q(\tau,h)$
	satisfies the first-order optimality condition
	$F(q)+h(q-m)=\tau,$
	where
	$m=\mathbb E(Y).$
	Consequently,
	$q_Q(\tau,h)$
	coincides with the interpolated quantile defined in subsection~\ref{subsec:oneone}.
	In particular, if
	\[
	h=h(\tau,\tau')
	=
	\frac{\tau-\tau'}
	{q(\tau')-m},
	\]
	then
	$q_Q(\tau,h(\tau,\tau'))=q(\tau')$
	and
	$F(q_Q(\tau,h(\tau,\tau')))=\tau'.$
\end{proposition}

\begin{proof}
	
	Define
	\[
	L(q)
	=
	\mathbb E[\rho_\tau(Y-q)]
	+
	\frac h2
	\mathbb E[(Y-q)^2].
	\]
	Differentiating with respect to \(q\) yields
	$
	\frac{d}{dq}
	\mathbb E[\rho_\tau(Y-q)]
	=
	F(q)-\tau,
	$
	and
	$
	\frac{d}{dq}
	\frac h2
	\mathbb E[(Y-q)^2]
	=
	h(q-m).
	$
	Hence
	$
	L'(q)
	=
	F(q)-\tau+h(q-m).
	$
	The first-order optimality condition
	$
	L'(q)=0
	$
	becomes
	$
	F(q)+h(q-m)=\tau.
	$
	By Proposition~\ref{prop:parametrization}, this equation admits a unique solution.
	Therefore, the minimizer of the variational problem coincides with the interpolated quantile introduced previously.
	If
	$
	h=h(\tau,\tau'),
	$
	Proposition~\ref{prop:parametrization} further implies
	$
	q_Q(\tau,h(\tau,\tau'))=q(\tau')
	$
	and
	$
	F(q_Q(\tau,h(\tau,\tau')))=\tau'.
	$
	
\end{proof}
\noindent
Therefore, the interpolation parameter
$
h
$
admits two equivalent interpretations:
as a probability-space deformation parameter linking the probability levels
$
\tau
$
and
$
\tau',
$
and as a regularization parameter in a penalized quantile optimization problem.

\subsection{Asymptotic normality}

\noindent
At the empirical level, given observations
$
y_1,\dots,y_n,
$
the empirical quadratic interpolation estimator is defined by
\[
\hat q_Q(\tau,h)
=
\arg\min_{q\in\mathbb R}
\left\{
\sum_{i=1}^n
\rho_\tau(y_i-q)
+
\frac h2
\sum_{i=1}^n
(y_i-q)^2
\right\}.
\]
The corresponding empirical interpolation probability level is
$
\hat p(\tau,h)
=
\hat F(\hat q_Q(\tau,h)),
$
where
$
\hat F
$
denotes the empirical distribution function.

\begin{theorem}
	\label{thm:quadratic-clt}
	
	Assume that the distribution function \(F\) admits a continuous density
	\(f\) satisfying
	$
	f(q_Q(\tau,h))>0,
	$
	and that
	$
	\mathbb E(Y^2)<\infty.
	$
	Let $h\ge0$ be fixed, and define
	\[
	\hat q_Q(\tau,h)
	=
	\arg\min_{q\in\mathbb R}
	\left\{
	\sum_{i=1}^{n}
	\rho_\tau(Y_i-q)
	+
	\frac h2
	\sum_{i=1}^{n}
	(Y_i-q)^2
	\right\},
	\]
	where
	$
	Y_1,\ldots,Y_n
	$
	are independent copies of \(Y\).
	Then
	\[
	\sqrt n
	\Bigl(
	\hat q_Q(\tau,h)
	-
	q_Q(\tau,h)
	\Bigr)
	\overset d\longrightarrow
	\mathcal N
	\!\left(
	0,
	\sigma_Q^2(\tau,h)
	\right),
	\]
	where
	\[
	\sigma_Q^2(\tau,h)
	=
	\frac{
		\mathbb E
		\Bigl[
		\bigl(
		\psi_\tau(Y-q_Q(\tau,h))
		+
		h(Y-q_Q(\tau,h))
		\bigr)^2
		\Bigr]
	}
	{
		\bigl(
		f(q_Q(\tau,h))+h
		\bigr)^2
	}.
	\]
	
\end{theorem}

\begin{proof}
	The asymptotic analysis is performed for a fixed regularization parameter $h$.
	The estimator $\hat q_Q(\tau,h)$ is the unique solution of the empirical
	estimating equation
	\[
	\frac1n\sum_{i=1}^n
	\Bigl(
	\psi_\tau(Y_i-q)
	+
	h(Y_i-q)
	\Bigr)
	=0.
	\]
	Since the score contains the discontinuous function
	$\psi_\tau$, the usual smooth M-estimation arguments do not
	apply directly. The asymptotic normality is established in
	Appendix~\ref{app:proof-quadratic-clt} using empirical-process
	methods for non-smooth M-estimators.
	The appendix shows that
	\[
	\sqrt n
	\Bigl(
	\hat q_Q(\tau,h)
	-
	q_Q(\tau,h)
	\Bigr)
	=
	\frac1{f(q_Q(\tau,h))+h}
	\frac1{\sqrt n}
	\sum_{i=1}^n
	\Bigl(
	\psi_\tau(Y_i-q_Q(\tau,h))
	+
	h(Y_i-q_Q(\tau,h))
	\Bigr)
	+o_p(1).
	\]
	Applying the central limit theorem yields
	\[
	\sqrt n
	\Bigl(
	\hat q_Q(\tau,h)
	-
	q_Q(\tau,h)
	\Bigr)
	\overset d\longrightarrow
	\mathcal N
	\!\left(
	0,
	\sigma_Q^2(\tau,h)
	\right),
	\]
	with
	\[
	\sigma_Q^2(\tau,h)
	=
	\frac{
		\mathbb E
		\!\left[
		\bigl(
		\psi_\tau(Y-q_Q(\tau,h))
		+h(Y-q_Q(\tau,h))
		\bigr)^2
		\right]
	}
	{
		\bigl(f(q_Q(\tau,h))+h\bigr)^2
	}.
	\]
\end{proof}

\begin{remark}
	
	When
	$
	h=0,
	$
	the quadratic interpolation estimator reduces to the classical quantile
	estimator. Since
	$
	q_Q(\tau,0)=q(\tau),
	$
	and
	$
	\mathbb E[\psi_\tau(Y-q(\tau))^2]
	=
	\tau(1-\tau),
	$
	Theorem~\ref{thm:quadratic-clt} becomes
	\[
	\sqrt n
	\bigl(
	\hat q(\tau)-q(\tau)
	\bigr)
	\overset d\longrightarrow
	\mathcal N
	\left(
	0,
	\frac{\tau(1-\tau)}
	{f(q(\tau))^2}
	\right),
	\]
	which is the classical asymptotic variance of the empirical quantile
	estimator.
\end{remark}

\subsection{Probability‑space deformation}

\begin{proposition}
	\label{prop:h-tau}
	Fix
	$
	\tau' \in (0,1),
	$
	$
	\tau' \neq F(m),
	$
	and consider the mapping
	$
	\tau \longmapsto h(\tau,\tau'),
	$
	defined for those values of \(\tau\) such that \(\tau'\) lies between
	\(\tau\) and \(F(m)\).
	Then
	\[
	h(\tau,\tau')
	=
	\frac{\tau-\tau'}{q(\tau')-m},
	\]
	so that \(h(\tau,\tau')\) is an affine function of \(\tau\). Moreover,
	\[
	\frac{\partial h}{\partial\tau}
	=
	\frac{1}{q(\tau')-m}.
	\]
	Consequently,
	$
	\tau' > F(m)
	\Longrightarrow
	q(\tau') > m,
	$
	and the mapping
	$
	\tau \longmapsto h(\tau,\tau')
	$
	is strictly increasing.
	Similarly,
	$
	\tau' < F(m)
	\Longrightarrow
	q(\tau') < m,
	$
	and the mapping
	$
	\tau \longmapsto h(\tau,\tau')
	$
	is strictly decreasing.
	Finally,
	$
	h(\tau',\tau')=0.
	$
	
\end{proposition}

\begin{proof}
	
	The identity
	\[
	h(\tau,\tau')
	=
	\frac{\tau-\tau'}{q(\tau')-m}
	\]
	follows directly from the definition of the interpolation parameter.
	Since \(q(\tau')\) and \(m\) are fixed, \(h(\tau,\tau')\) is affine in
	\(\tau\), and
	\[
	\frac{\partial h}{\partial\tau}
	=
	\frac{1}{q(\tau')-m}.
	\]
	If \(\tau' > F(m)\), then \(q(\tau') > m\), so that
	$
	\frac{\partial h}{\partial\tau}>0,
	$
	which implies that \(h(\tau,\tau')\) is strictly increasing.
	If \(\tau' < F(m)\), then \(q(\tau') < m\), so that
	$
	\frac{\partial h}{\partial\tau}<0,
	$
	which implies that \(h(\tau,\tau')\) is strictly decreasing.
	The identity
	$
	h(\tau',\tau')=0
	$
	is immediate from the definition.
\end{proof}

\noindent
Consequently, fixing
$
(\tau,\tau')
$
determines a unique interpolation parameter
$
h(\tau,\tau'),
$
while fixing
$
(\tau,h)
$
determines a unique effective probability level
$
\tau'
=
F(q(\tau,h)).
$
When
$
h=0,
$
the interpolation equation reduces to the classical quantile equation
$
F(q)=\tau,
$
so that
$
q(\tau,0)=q(\tau).
$
As the interpolation parameter increases, the interpolated quantile is continuously deformed toward the mean \(m\).
Thus, the interpolation equation defines a continuous deformation mechanism in probability space linking classical quantiles to the central location of the distribution.

\subsection{Monotone deformation}
\label{sec:monotone}

The parametrization introduced in subsection~\ref{subsec:oneone} induces a continuous deformation of probability levels through the interpolation parameter
$
h.
$
For every fixed pair
$
(\tau,h),
$
define the effective interpolation probability level
$
p(\tau,h)
=
F(q_Q(\tau,h)),
$
where
$
q_Q(\tau,h)
$
denotes the solution of the interpolation equation
$
F(q)+h(q-m)=\tau.
$
By Proposition~\ref{prop:parametrization}, fixing
$
(\tau,h)
$
determines a unique probability level
$
p(\tau,h).
$
Moreover, when
$
h=0,
$
the interpolation equation reduces to the classical quantile equation, so that
\[
q_Q(\tau,0)=q(\tau),
\qquad
p(\tau,0)=\tau.
\]
As the interpolation parameter increases, the interpolated quantile is progressively transported toward the mean
$
m=\mathbb E(Y),
$
suggesting that the corresponding probability level should converge toward
$
F(m).
$
The following proposition characterizes the monotone structure of this probability-space deformation.

\begin{proposition}
	\label{prop:probability-interpolation}
	
	Assume that the distribution function $F$ admits a continuous density $f$ satisfying
	$
	f(x)>0
	$
	for every
	$
	x\in
	[\min\{q(\tau),m\},
	\max\{q(\tau),m\}].
	$
	Then, for every fixed
	$
	\tau\in(0,1),
	$
	the mapping
	$
	h\mapsto q_Q(\tau,h)
	$
	is continuous on
	$
	[0,+\infty).
	$
	Moreover:
	
	\begin{itemize}
		
		\item
		if
		$
		\tau<F(m),
		$
		then both
		$
		q_Q(\tau,h)
		$
		and
		$
		p(\tau,h)=F(q_Q(\tau,h))
		$
		are increasing functions of \(h\);
		
		\item
		if
		$
		\tau>F(m),
		$
		then both
		$
		q_Q(\tau,h)
		$
		and
		$
		p(\tau,h)=F(q_Q(\tau,h))
		$
		are decreasing functions of \(h\).
	\end{itemize}
	Furthermore,
	\[
	\lim_{h\to+\infty} q_Q(\tau,h)=m,
	\qquad
	\lim_{h\to+\infty} p(\tau,h)=F(m).
	\]
	
\end{proposition}

\begin{remark}
	Proposition~\ref{prop:probability-interpolation} describes the deformation path at the population level through
	$
	p(\tau,h)=F\bigl(q_Q(\tau,h)\bigr).
	$
	Since \(F\) is continuous and \(h \mapsto q_Q(\tau,h)\) is continuous, the path \(h \mapsto p(\tau,h)\) is itself continuous.
	At the empirical level, one may similarly define
	$
	\hat p(\tau,h)
	=
	\hat F\bigl(\hat q_Q(\tau,h)\bigr),
	$
	where \(\hat F\) denotes the empirical distribution function. In contrast with \(F\), the empirical distribution function \(\hat F\) is a step function. Consequently, the empirical deformation path \(h \mapsto \hat p(\tau,h)\) is generally piecewise constant and may exhibit jumps. Thus, Proposition~\ref{prop:probability-interpolation} should be interpreted as describing the continuous population deformation, while its empirical analogue provides a staircase approximation of this path.
\end{remark}

\begin{proof}
	The interpolation equation satisfies
	\[
	F(q_Q(\tau,h))
	+
	h(q_Q(\tau,h)-m)
	=
	\tau.
	\]
	By Proposition~\ref{prop:parametrization}, for every \(h\ge0\),
	the interpolated quantile \(q_Q(\tau,h)\) corresponds to a probability
	level between \(\tau\) and \(F(m)\). Hence
	\[
	q_Q(\tau,h)
	\in
	[\min\{q(\tau),m\},
	\max\{q(\tau),m\}],
	\]
	for all \(h\ge0\).
	Since \(f\) is strictly positive on this interval,
	$
	f(q_Q(\tau,h))>0.
	$
	Therefore
	$
	f(q_Q(\tau,h))+h>0,
	$
	and the implicit function theorem applies.
	Differentiating with respect to \(h\) yields
	\[
	\frac{\partial q_Q}{\partial h}
	=
	-
	\frac{
		q_Q(\tau,h)-m
	}{
		f(q_Q(\tau,h))+h
	}.
	\]
	Hence, the sign of
	$
	\frac{\partial q_Q}{\partial h}
	$
	is opposite to the sign of
	$
	q_Q(\tau,h)-m.
	$
	If
	$
	\tau<F(m),
	$
	then
	$
	q_Q(\tau,h)<m,
	$
	which implies
	$
	\frac{\partial q_Q}{\partial h}>0.
	$
	If
	$
	\tau>F(m),
	$
	then
	$
	q_Q(\tau,h)>m,
	$
	which implies
	$
	\frac{\partial q_Q}{\partial h}<0.
	$
	Since \(F\) is increasing, the same monotonicity properties hold for
	$
	p(\tau,h)=F(q_Q(\tau,h)).
	$
	Finally, the interpolation equation implies
	\[
	q_Q(\tau,h)-m
	=
	\frac{\tau-p(\tau,h)}{h}.
	\]
	Since
	$
	0\le p(\tau,h)\le1,
	$
	the right-hand side converges to zero as
	$
	h\to+\infty.
	$
	Hence
	$
	q_Q(\tau,h)\to m.
	$
	By continuity of \(F\),
	$
	p(\tau,h)=F(q_Q(\tau,h))\to F(m).
	$
\end{proof}
\noindent
Proposition~\ref{prop:probability-interpolation} shows that the interpolation parameter
$
h
$
generates a monotone transport of probability levels toward the central probability value
$
F(m).
$
Two regimes naturally arise.

\medskip

\noindent
{\bf Increasing regime.}
If
$
\tau<F(m),
$
then
$
p(\tau,h)
$
increases continuously from
$
\tau
$
toward
$
F(m).
$
Consequently, for every probability level
$
\tau<\tau'<F(m),
$
there exists a unique interpolation parameter
$
h(\tau,\tau')
$
such that
$
p(\tau,h(\tau,\tau'))=\tau'.
$
By Proposition~\ref{prop:parametrization},
\[
q_Q(\tau,h(\tau,\tau'))=q(\tau').
\]
{\bf Decreasing regime.}
If
$
\tau>F(m),
$
then
$
p(\tau,h)
$
decreases continuously from
$
\tau
$
toward
$
F(m).
$
Hence, for every
$
F(m)<\tau'<\tau,
$
there exists a unique interpolation parameter
$
h(\tau,\tau')
$
such that
$
p(\tau,h(\tau,\tau'))=\tau'.
$
Again, Proposition~\ref{prop:parametrization} implies
\[
q_Q(\tau,h(\tau,\tau'))=q(\tau').
\]
Therefore, the interpolation parameter
$
h
$
defines a continuous path in probability space linking quantile levels through the interpolation equation.
At the empirical level, define
$
\hat p(\tau,h)
=
\hat F(\hat q_Q(\tau,h)),
$
where
$
\hat q_Q(\tau,h)
$
denotes the empirical quadratic interpolation estimator and
$
\hat F
$
the empirical distribution function.
Given a target probability level
$
\tau'
$
located between
$
\tau
$
and
$
F(m),
$
the interpolation parameter
$
h(\tau,\tau')
$
may be estimated numerically as the solution of
$
\hat F(\hat q_Q(\tau,h))
=
\tau'.
$
This inversion problem is one-dimensional and can be solved efficiently by numerical root finding.

\subsection{Numerical illustrations}

We now illustrate the probability-space interpolation mechanism in two complementary situations using a single asymmetric Laplace sample.
More precisely, one fixed sample is generated from an asymmetric Laplace distribution with parameter
$
\tau_0=0.8,
$
and both interpolation regimes are subsequently studied using the same observations.
For this sample, the empirical probability level associated with the sample mean is
$
\hat F(\bar Y)\approx0.38.
$
Consequently, the interpolation behavior depends on whether the initial probability level
\(
\tau
\)
lies above or below
\(
\hat F(\bar Y)
\).
Recall that the interpolated probability level is defined by
$
\hat p(\tau,h)
=
\hat F(\hat q_Q(\tau,h)),
$
where
$
\hat q_Q(\tau,h)
$
denotes the quadratically regularized quantile estimator.
According to Proposition~\ref{prop:probability-interpolation}, the mapping
$
h\mapsto \hat p(\tau,h)
$
is monotone and converges toward
\(
\hat F(\bar Y)
\)
as
\(
h\to\infty
\).
The left panel of Figure~\ref{fig:probability_interpolation_combined}
illustrates the decreasing regime
$
\tau>\hat F(\bar Y),
$
while the right panel illustrates the increasing regime
$
\tau<\hat F(\bar Y).
$
In both experiments, we also compare two estimators of the interpolation parameter
\(
h(\tau,\tau')
\)
associated with a target probability level
\(
\tau'
\):
$
\hat h_{\mathrm{root}}(\tau,\tau')
=
\text{solution of }
\hat F(\hat q_Q(\tau,h))
=
\tau',
$
obtained by numerical root finding, and the explicit plug-in estimator
\[
\hat h_{\mathrm{explicit}}(\tau,\tau')
=
\frac{\tau-\tau'}
{\hat q(\tau')-\bar Y},
\]
where
\(
\hat q(\tau')
\)
denotes the empirical quantile estimator.
The numerical results show an excellent agreement between these two approaches.

\paragraph{Decreasing regime: \(\tau > \hat F(\bar Y)\).}

We first consider the upper-tail interpolation regime with
$
\tau=0.8.
$
Since
$
\hat F(\bar Y)\approx0.38,
$
one has
$
\tau=0.8>\hat F(\bar Y),
$
and therefore the interpolation produces a monotone decreasing deformation of the probability level from
\(
0.8
\)
toward
\(
0.38
\).
The left panel of Figure~\ref{fig:probability_interpolation_combined}
displays the evolution of
$
\hat p(\tau,h)
=
\hat F(\hat q_Q(\tau,h))
$
as a function of the regularization parameter
\(
h
\).
At
\(
h=0
\),
the regularized estimator coincides with the empirical quantile estimator, yielding
$
\hat p(0.8,0)=0.8.
$
As
\(
h
\)
increases, the quadratic penalty progressively shrinks the estimator toward the sample mean, and the associated probability level decreases continuously toward
\(
\hat F(\bar Y)\approx0.38
\).
To illustrate the inverse interpolation problem, we consider the target levels
$
\tau'\in\{0.75,0.55\}.
$
Both values belong to the reachable interval
$
[\hat F(\bar Y),\tau]
\approx
[0.38,0.8],
$
and therefore admit interpolation solutions.
The corresponding root-finding estimators are
$
\hat h_{\mathrm{root}}(0.8,0.75)
\approx
0.014559,
$
and
$
\hat h_{\mathrm{root}}(0.8,0.55)
\approx
0.132758.
$
The associated interpolated estimators satisfy
$
\hat q_Q(0.8,\hat h_{\mathrm{root}}(0.8,0.75))
\approx
-0.341923,
$
and
$
\hat q_Q(0.8,\hat h_{\mathrm{root}}(0.8,0.55))
\approx
-1.905605.
$
Moreover,
$
\hat p(0.8,\hat h_{\mathrm{root}}(0.8,0.75))
=
0.75,
$
and
$
\hat p(0.8,\hat h_{\mathrm{root}}(0.8,0.55))
=
0.55.
$
The corresponding explicit plug-in estimators are
$
\hat h_{\mathrm{explicit}}(0.8,0.75)
\approx
0.014507,
$
and
$
\hat h_{\mathrm{explicit}}(0.8,0.55)
\approx
0.132765.
$
The near-perfect agreement between the two estimators numerically confirms the validity of the interpolation identity.

\paragraph{Increasing regime: \(\tau < \hat F(\bar Y)\).}

We next consider the lower-tail interpolation regime with
$
\tau=0.1.
$
Since
$
\tau=0.1<\hat F(\bar Y)\approx0.38,
$
the interpolation now generates a monotone increasing deformation of the probability level toward
\(
\hat F(\bar Y)
\).
The right panel of Figure~\ref{fig:probability_interpolation_combined}
shows the evolution of
$
\hat p(\tau,h)
=
\hat F(\hat q_Q(\tau,h))
$
as a function of
\(
h
\).
At
\(
h=0
\),
one recovers the empirical quantile estimator:
$
\hat p(0.1,0)=0.1.
$
As
\(
h
\)
increases, the quadratic penalty pulls the estimator toward the sample mean, producing a continuous increase toward
\(
\hat F(\bar Y)\approx0.38
\).
We consider the target probability levels
$
\tau'\in\{0.25,0.30\}.
$
Since both values belong to the reachable interval
$
[\tau,\hat F(\bar Y)]
\approx
[0.1,0.38],
$
the corresponding interpolation equations admit solutions.
The root-finding interpolation estimators are
$
\hat h_{\mathrm{root}}(0.1,0.25)
\approx
0.076075,
$
and
$
\hat h_{\mathrm{root}}(0.1,0.30)
\approx
0.187780.
$
The associated interpolated estimators are
$
\hat q_Q(0.1,\hat h_{\mathrm{root}}(0.1,0.25))
\approx
-5.748222,
$
and
$
\hat q_Q(0.1,\hat h_{\mathrm{root}}(0.1,0.30))
\approx
-4.856126.
$
Furthermore,
$
\hat p(0.1,\hat h_{\mathrm{root}}(0.1,0.25))
=
0.25,
$
and
$
\hat p(0.1,\hat h_{\mathrm{root}}(0.1,0.30))
=
0.30.
$
The corresponding explicit plug-in estimators are
$
\hat h_{\mathrm{explicit}}(0.1,0.25)
\approx
0.076546,
$
and
$
\hat h_{\mathrm{explicit}}(0.1,0.30)
\approx
0.187354.
$
Again, the two approaches produce nearly identical interpolation parameters.
Altogether, these experiments numerically confirm both monotonicity regimes established in Proposition~\ref{prop:probability-interpolation}, and demonstrate that the interpolation parameter can be accurately recovered either by numerical inversion or through the explicit plug-in approximation.
Further empirical illustrations with real data are provided in Section~\ref{sec:cac40}.

\begin{figure}[!htbp]
	\centering
	\includegraphics[width=\textwidth]
	{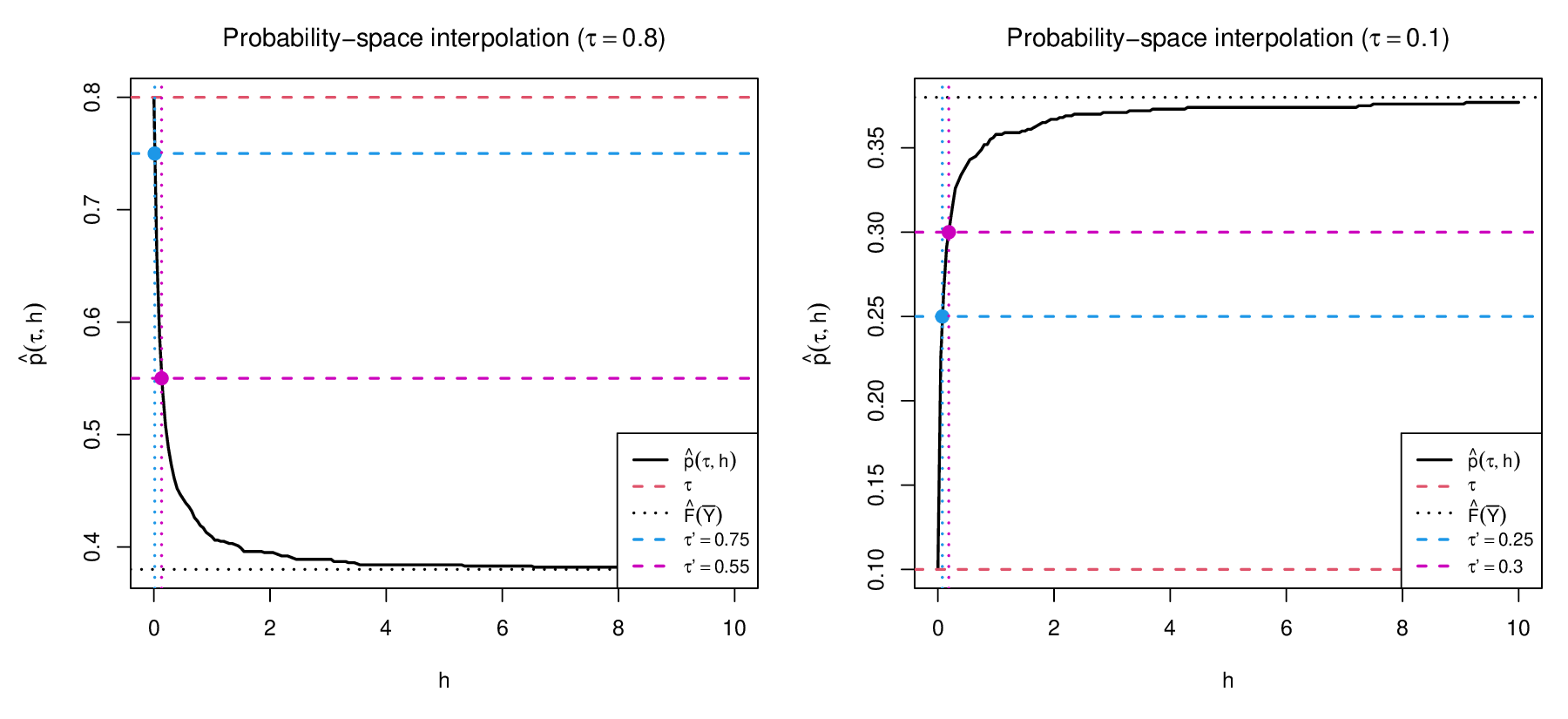}
	\caption{
		Probability-space interpolation for the quadratically regularized quantile estimator using a single asymmetric Laplace sample with
		\(
		\hat F(\bar Y)\approx0.38
		\).
		Left panel: decreasing regime
		\(
		\tau=0.8>\hat F(\bar Y)
		\),
		for which the interpolated probability level decreases monotonically toward
		\(
		\hat F(\bar Y)
		\).
		Right panel: increasing regime
		\(
		\tau=0.1<\hat F(\bar Y)
		\),
		for which the interpolated probability level increases monotonically toward
		\(
		\hat F(\bar Y)
		\).
		The colored horizontal lines correspond to the target probability levels
		\(
		\tau'
		\),
		while the vertical dashed lines indicate the associated interpolation parameters.
	}
	\label{fig:probability_interpolation_combined}
\end{figure}

\subsection{Distributional structure of the asymptotic variance}

As will be shown by the numerical experiments of Section~\ref{sec:efficiency},
the optimal interpolation strength depends strongly on the underlying
distribution. To better understand this phenomenon, we examine the 
structure of the asymptotic variance appearing in
Theorem~\ref{thm:quadratic-clt}.
Let
$
q=q_Q(\tau,h)
$
denote the unique solution of the interpolation equation
\begin{equation}
	\label{eq:interpolation-equation}
	F(q)+h(q-m)=\tau,
	\qquad
	m=\mathbb{E}[Y].
\end{equation}
Using \eqref{eq:interpolation-equation}, the score function can be
rewritten as
\[
\psi_\tau(Y-q)+h(Y-q)
=
\bigl(F(q)-\mathbf 1_{\{Y<q\}}\bigr)
+
h(Y-m).
\]
Consequently, the asymptotic variance of the quadratic interpolation
estimator admits the following representation.

\begin{proposition}
	\label{prop:variance-structure}
	Let
	\[
	T(q)
	=
	\mathbb{E}
	\!\left[
	(Y-m)\mathbf 1_{\{Y<q\}}
	\right]
	\]
	denote the truncated first moment and let
	$
	\sigma_Y^2=\operatorname{Var}(Y).
	$
	Then
	\[
	\sigma_Q^2(\tau,h)
	=
	\frac{
		F(q)\bigl(1-F(q)\bigr)
		+
		h^2\sigma_Y^2
		-
		2h\,T(q)
	}
	{\bigl(f(q)+h\bigr)^2},
	\]
	where \(q=q_Q(\tau,h)\) is defined by
	\eqref{eq:interpolation-equation}.
\end{proposition}

\begin{proof}
	Since
	\[
	\psi_\tau(Y-q)+h(Y-q)
	=
	\bigl(F(q)-\mathbf 1_{\{Y<q\}}\bigr)
	+
	h(Y-m),
	\]
	the numerator of the asymptotic variance is
	$
	\operatorname{Var}
	\!\left(
	F(q)-\mathbf 1_{\{Y<q\}}
	+
	h(Y-m)
	\right).
	$
	Using
	$
	\operatorname{Var}
	\!\left(
	F(q)-\mathbf 1_{\{Y<q\}}
	\right)
	=
	F(q)\bigl(1-F(q)\bigr),
	$
	and observing that
	\[
	\begin{aligned}
		\operatorname{Cov}
		\!\left(
		F(q)-\mathbf1_{\{Y<q\}},
		Y-m
		\right)
		&=
		\mathbb E\!\left[
		\bigl(F(q)-\mathbf1_{\{Y<q\}}\bigr)(Y-m)
		\right] \\
		&=
		-\mathbb E\!\left[
		(Y-m)\mathbf1_{\{Y<q\}}
		\right] \\
		&=
		-T(q),
	\end{aligned}
	\]
	where we used the fact that
	\(\mathbb E(Y-m)=0\),
	gives
	\[
	\mathbb E
	\Big[
	\big(
	\psi_\tau(Y-q)+h(Y-q)
	\big)^2
	\Big]
	=
	F(q)\bigl(1-F(q)\bigr)
	+
	h^2\sigma_Y^2
	-
	2h\,T(q).
	\]
	The result follows from Theorem~\ref{thm:quadratic-clt}.
\end{proof}
Proposition~\ref{prop:variance-structure} reveals explicitly how the
underlying distribution affects the asymptotic variance.
Three distributional quantities appear:
\begin{enumerate*}[label*=\arabic*)]
	\item the local density \(f(q)\),
	\item the global variance \(\sigma_Y^2\),
	\item the truncated first moment
	\(
	T(q)
	=
	\mathbb E[(Y-m)\mathbf 1_{\{Y<q\}}].
	\)
\end{enumerate*}
The corresponding asymptotic variance of the empirical quantile
estimator targeting the same value \(q\) is
\[
\sigma_{\mathrm{emp}}^2
=
\frac{
	F(q)\bigl(1-F(q)\bigr)
}
{f(q)^2}.
\]
Hence the relative efficiency ratio becomes
\[
R(\tau,h)
=
\frac{
	F(q)(1-F(q))
	+
	h^2\sigma_Y^2
	-
	2h\,T(q)
}
{
	F(q)(1-F(q))
}
\,
\frac{
	f(q)^2
}
{
	(f(q)+h)^2
}.
\]
For a fixed quantile level \(\tau\), the interpolated quantile
\(q=q_Q(\tau,h)\) varies with \(h\).
Differentiating \eqref{eq:interpolation-equation} yields
\[
\frac{dq}{dh}
=
-\frac{q-m}{f(q)+h}.
\]
Therefore both the numerator and denominator of
\(\sigma_Q^2(\tau,h)\) depend on \(h\) through the induced variation
of \(q\).
The variance-optimal interpolation strength
\[
h_{\mathrm{var}}^\star(\tau)
=
\arg\min_{h\ge0}\sigma_Q^2(\tau,h)
\]
This variance-optimal interpolation strength should not be confused with the efficiency-optimal interpolation strength introduced in Section~\ref{sec:efficiency}, which minimizes the relative efficiency ratio \(R(\tau,h)\). Since the reference variance \(\sigma_{\mathrm{emp}}^2(\tau_h)\) also depends on \(h\), the two optimization criteria are generally not equivalent.
$h_{\mathrm{var}}^\star(\tau)$
is consequently determined by an implicit first-order condition
involving \(f\), \(f'\), \(T\), and \(q\).
In general, no closed-form expression appears to be available.
Nevertheless, Proposition~\ref{prop:variance-structure} provides a
useful qualitative interpretation of the numerical results reported in
Figures~\ref{fig:contour-comparison} and
\ref{fig:optimal-h-comparison}.
The optimal amount of interpolation is governed by a balance between
the density \(f(q)\), which controls the local sensitivity of the
quantile, and the global distributional quantities
\(\sigma_Y^2\) and \(T(q)\).
This explains why the Gaussian and Laplace distributions exhibit
markedly different optimal interpolation profiles despite possessing
similar symmetry properties.

\subsection{Efficiency comparisons}
\label{sec:efficiency}

\subsubsection{Comparison with the empirical quantile estimator}

Theorem~\ref{thm:quadratic-clt} provides the asymptotic variance of the quadratic interpolation estimator,
$
\sigma_Q^2(\tau,h).
$
A natural benchmark is the classical empirical quantile estimator corresponding to the same population target.
Let
$
q_Q(\tau,h)
$
denote the solution of
$
F(q)+h(q-m)=\tau,
$
and define
\[
\tau_h := F\!\left(q_Q(\tau,h)\right).
\]
The empirical quantile estimator targeting the quantile level \(\tau_h\) has asymptotic variance
\[
\sigma_{\mathrm{emp}}^2(\tau_h)
=
\frac{\tau_h(1-\tau_h)}
{f\!\left(q_Q(\tau,h)\right)^2}.
\]
To compare both procedures we consider the relative efficiency ratio
\[
R(\tau,h)
=
\frac{\sigma_Q^2(\tau,h)}
{\sigma_{\mathrm{emp}}^2(\tau_h)}.
\]
Values \(R(\tau,h)<1\) indicate that the quadratic interpolation estimator is asymptotically more efficient than the empirical quantile estimator targeting the same population quantile.
Two representative examples are considered below for \(\tau=0.9\).

\paragraph{Gaussian distribution.}

Consider the standard Gaussian distribution,
$
Y \sim \mathcal{N}(0,1).
$
Figure~\ref{fig:gaussian-efficiency} displays the asymptotic variances
\(\sigma_Q^2(\tau,h)\) and
\(\sigma_{\mathrm{emp}}^2(\tau_h)\)
as functions of \(h\), together with the efficiency ratio \(R(\tau,h)\).
The ratio decreases monotonically over the investigated interval.
For \(h=2\), the smallest value considered, we obtain
$
R(0.9,2) \approx 0.64.
$
Thus the quadratic interpolation estimator achieves an asymptotic variance approximately 36\% smaller than that of the corresponding empirical quantile estimator.
This behavior suggests that, for light-tailed distributions such as the Gaussian law, increasing the interpolation strength continuously improves asymptotic efficiency.

\paragraph{Laplace distribution.}

Consider the Laplace distribution with density
\[
f(x)=\frac12 e^{-|x|}.
\]
Figure~\ref{fig:laplace-efficiency} shows the analogous comparison.
Unlike the Gaussian case, the efficiency ratio is no longer monotone.
Instead, it possesses a clear minimum at
$
h_{\mathrm{eff}}^\star \approx 0.141,
$
for which
$
R(0.9,h_{\mathrm{eff}}^\star) \approx 0.697.
$
Hence the quadratic interpolation estimator achieves an asymptotic variance approximately 30\% smaller than that of the corresponding empirical quantile estimator.
For larger values of \(h\), the efficiency ratio increases and eventually exceeds one. In this regime the empirical quantile estimator becomes asymptotically preferable.
This phenomenon indicates that heavy-tailed distributions may admit an optimal interpolation strength. While moderate interpolation can substantially reduce asymptotic variance, excessively large values of $h$ eventually increase the relative efficiency ratio. The optimal value therefore reflects a balance between the benefits of interpolation and the loss of efficiency associated with overly strong regularization.

\begin{figure}[!htbp]
	\centering
	\includegraphics[width=\textwidth]{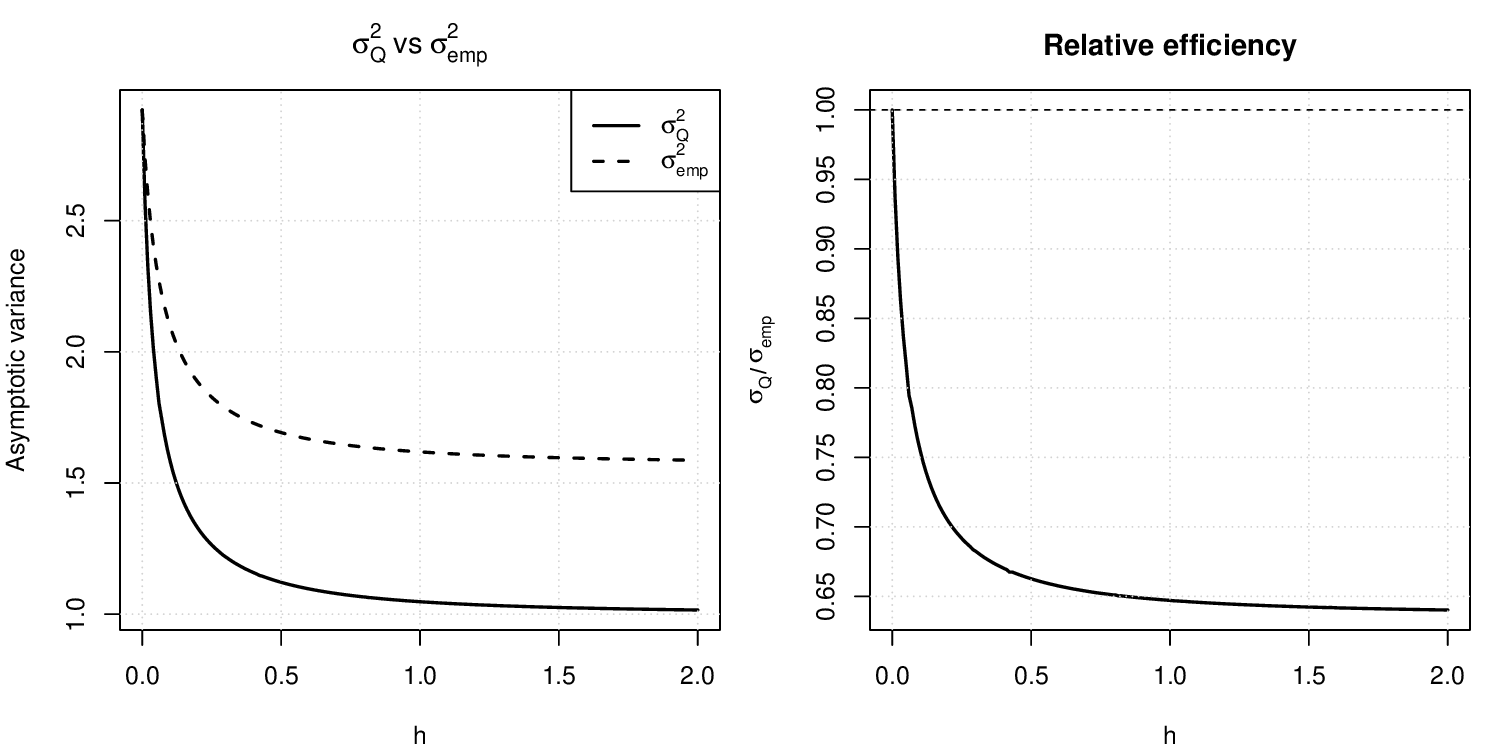}
	\caption{
		Comparison of asymptotic variances for the quadratic interpolation estimator and the empirical quantile estimator under the standard Gaussian distribution (\(\tau=0.9\)). Left: asymptotic variances. Right: relative efficiency ratio \(R(\tau,h)\).
	}
	\label{fig:gaussian-efficiency}
\end{figure}

\begin{figure}[!htbp]
	\centering
	\includegraphics[width=\textwidth]{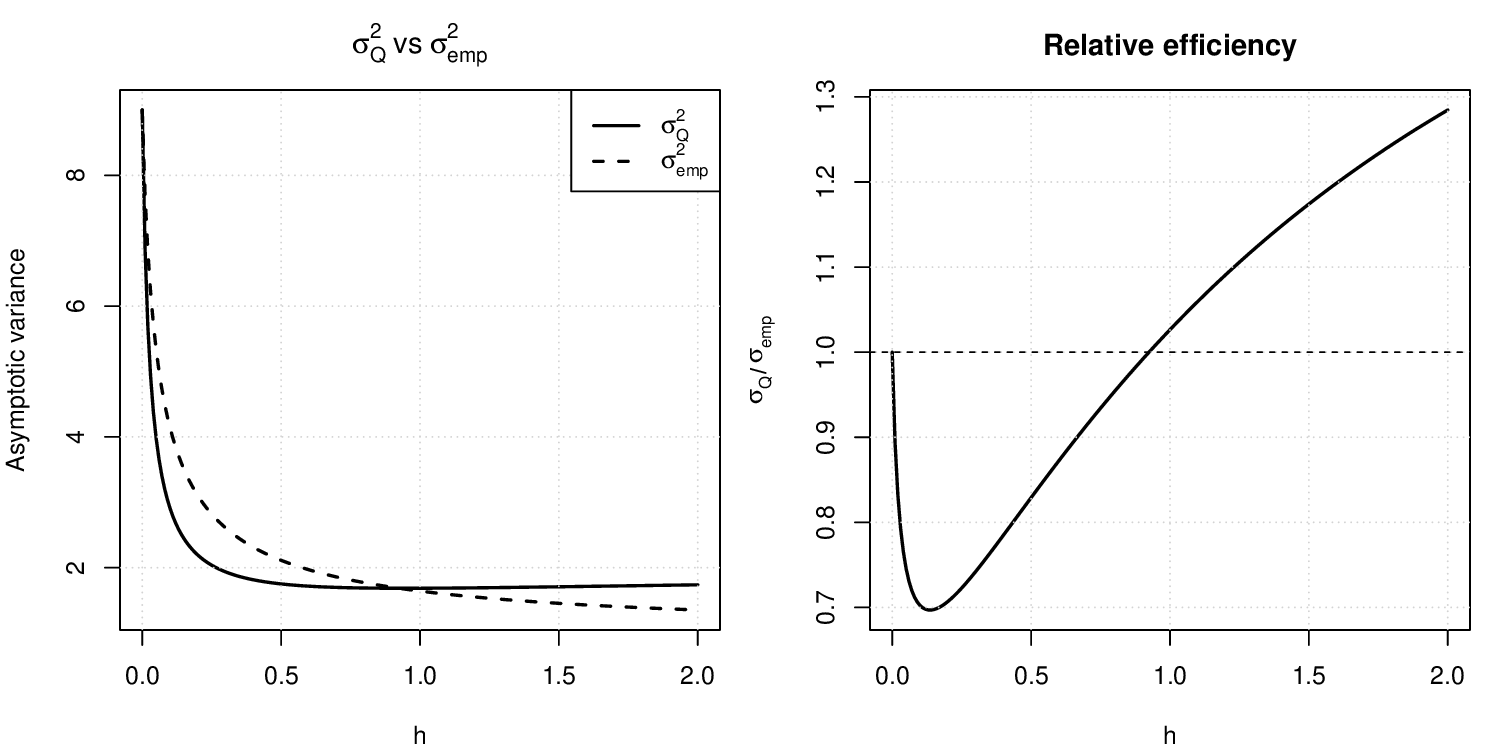}
	\caption{
		Comparison of asymptotic variances for the quadratic interpolation estimator and the empirical quantile estimator under the Laplace distribution (\(\tau=0.9\)). Left: asymptotic variances. Right: relative efficiency ratio \(R(\tau,h)\).
	}
	\label{fig:laplace-efficiency}
\end{figure}

\subsubsection{Influence of the quantile level}

The previous examples were restricted to the quantile level \(\tau=0.9\). To investigate the influence of the quantile level on the relative efficiency of the quadratic interpolation estimator, we now examine the efficiency ratio
\[
R(\tau,h)
=
\frac{\sigma_Q^2(\tau,h)}
{\sigma_{\mathrm{emp}}^2(\tau_h)}
\]
over a broad range of values of \(\tau\) and \(h\).
Figure~\ref{fig:contour-comparison} displays contour plots of \(R(\tau,h)\) for the Gaussian and Laplace distributions. The contour \(R(\tau,h)=1\) separates regions where the quadratic interpolation estimator is asymptotically more efficient than the empirical quantile estimator from regions where the opposite holds.

For the Gaussian distribution, the efficiency ratio remains strictly below one throughout the investigated parameter domain. Moreover, the contour plot is nearly invariant with respect to \(\tau\), indicating that the relative efficiency depends primarily on the interpolation parameter \(h\). The symmetry of the Gaussian distribution implies
$
R(\tau,h)=R(1-\tau,h),
$
which explains the symmetry of the contour plot about \(\tau=1/2\). As \(h\) increases, the efficiency ratio decreases monotonically, reaching values close to \(0.64\) for the largest values of \(h\) considered. Thus the quadratic interpolation estimator uniformly dominates the empirical quantile estimator over the entire range of quantile levels investigated.

The Laplace distribution exhibits a markedly different behavior. Although the symmetry relation
$
R(\tau,h)=R(1-\tau,h)
$
still holds, the dependence on \(\tau\) is much stronger. The contour plot reveals a nontrivial region where \(R(\tau,h)<1\), separated from a region where \(R(\tau,h)>1\) by the contour \(R(\tau,h)=1\). For each quantile level there exists an optimal interpolation strength \(h_{\mathrm{eff}}^\star(\tau)\) that minimizes the efficiency ratio. The gain is particularly pronounced for extreme quantiles, where values of \(R(\tau,h)\) below \(0.65\) are attained. In contrast, for central quantiles and sufficiently large values of \(h\), the efficiency ratio exceeds one and the empirical quantile estimator becomes asymptotically preferable.

These results show that the role of the quantile level depends strongly on the tail behavior of the underlying distribution. For light-tailed distributions such as the Gaussian law, the influence of \(\tau\) is negligible and increasing the interpolation strength systematically improves efficiency. For heavier-tailed distributions such as the Laplace law, the quantile level substantially affects the optimal choice of \(h\), and the largest efficiency gains occur when estimating extreme quantiles.

\begin{figure}[!htbp]
	\centering
	
	\begin{minipage}{0.49\textwidth}
		\centering
		\includegraphics[width=\textwidth]{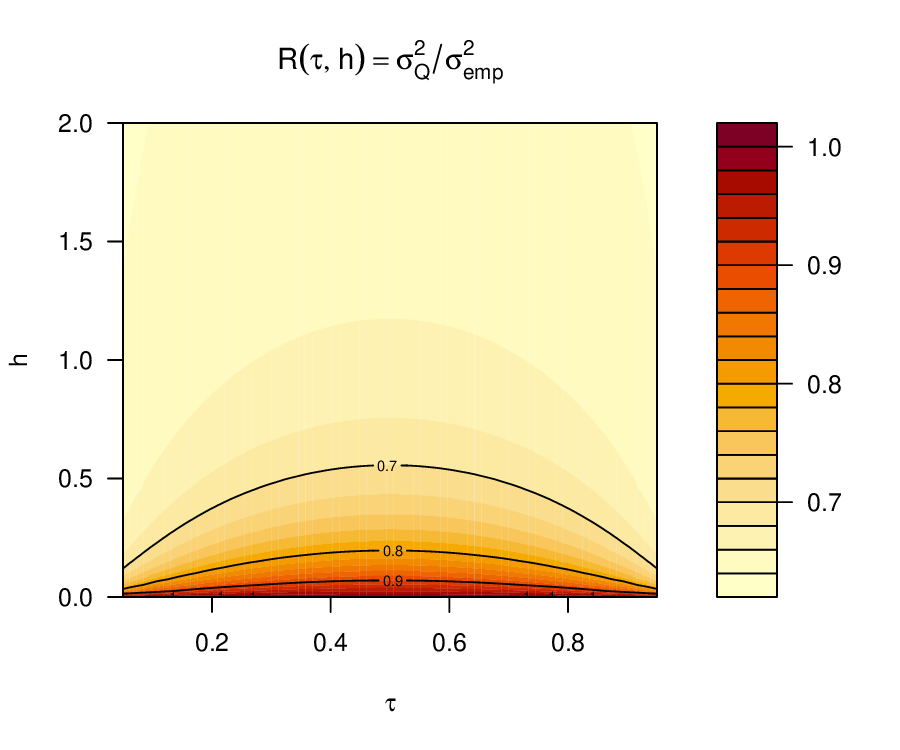}
		\subcaption{Gaussian distribution}
	\end{minipage}
	\hfill
	\begin{minipage}{0.49\textwidth}
		\centering
		\includegraphics[width=\textwidth]{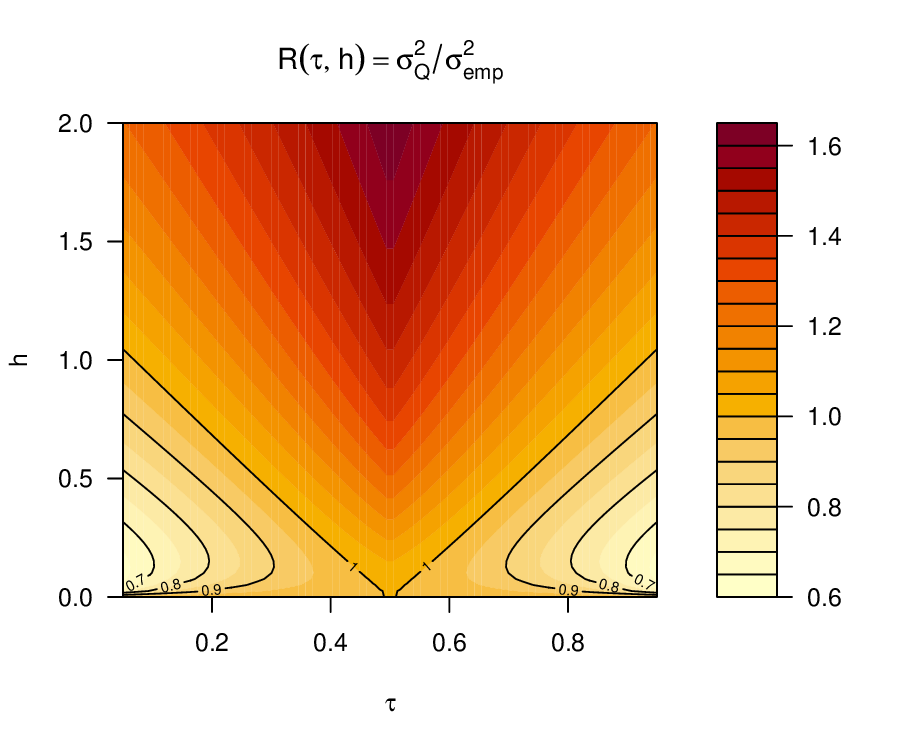}
		\subcaption{Laplace distribution}
	\end{minipage}
	
	\caption{
		Contour plots of the relative efficiency ratio
		$
		R(\tau,h)
		=
		\frac{\sigma_Q^2(\tau,h)}
		{\sigma_{\mathrm{emp}}^2(\tau_h)}.
		$
		The contour level \(R=1\) separates regions where the quadratic interpolation estimator is asymptotically more efficient than the empirical quantile estimator (\(R<1\)) from regions where the empirical quantile estimator is preferable (\(R>1\)). The Gaussian case exhibits uniform efficiency gains over the entire parameter domain, whereas the Laplace case displays a nontrivial efficiency boundary and an optimal interpolation strength depending on the quantile level.}
	\label{fig:contour-comparison}
\end{figure}

\paragraph{Optimal interpolation strength.}

The contour plots suggest that the role of the interpolation parameter $h$
depends strongly on the underlying distribution.
To quantify this effect, we define
\[
h_{\mathrm{eff}}^\star(\tau)
=
\arg\min_{h\ge0} R(\tau,h),
\]
and
\[
R^\star(\tau)
=
\min_{h\ge0} R(\tau,h).
\]
Figure~\ref{fig:optimal-h-comparison} displays the resulting optimal
interpolation strength $h_{\mathrm{eff}}^\star(\tau)$ together with the corresponding
minimal efficiency ratio $R^\star(\tau)$ for both the Gaussian and
Laplace distributions.

\begin{figure}[!htbp]
	\centering
	
	\includegraphics[width=0.95\textwidth]{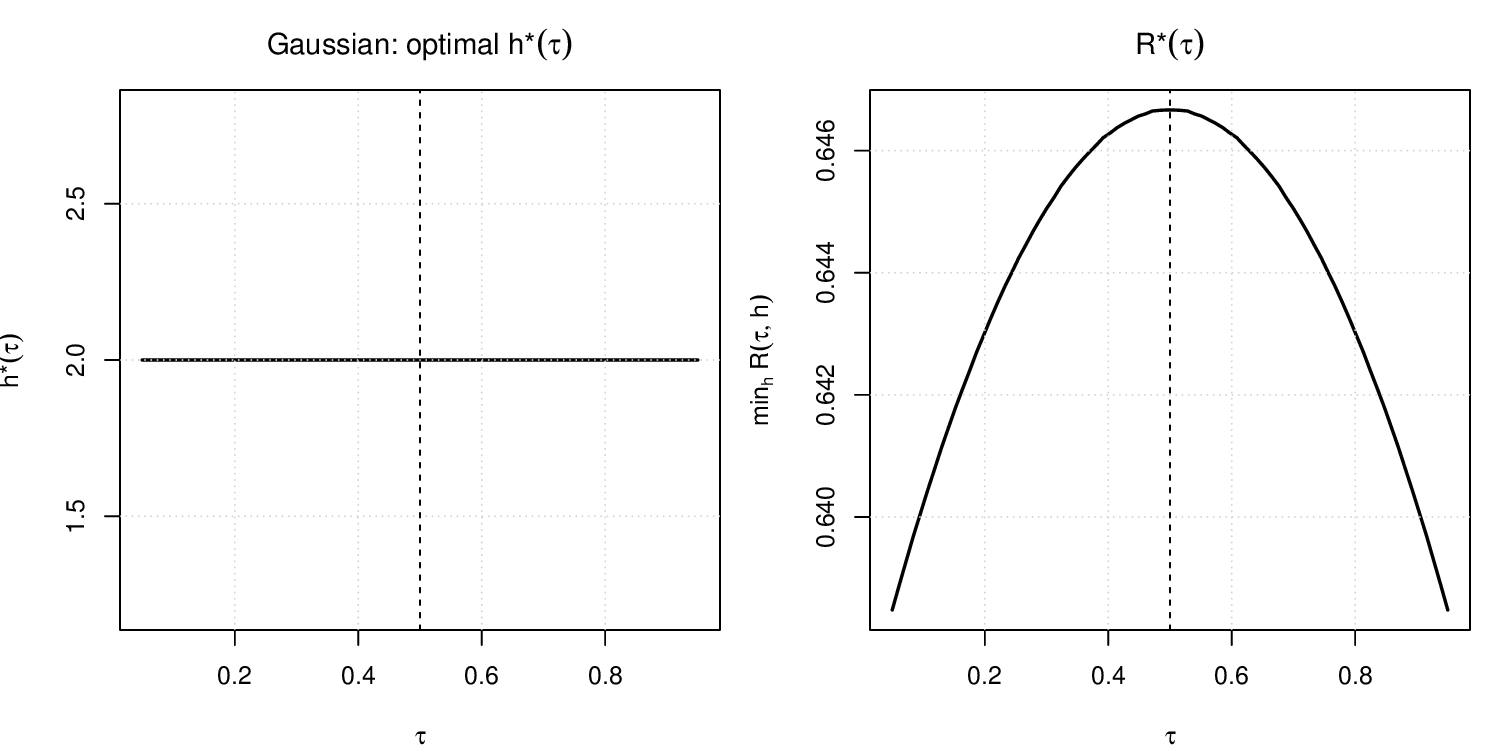}
	
	\vspace{0.3cm}
	
	\includegraphics[width=0.95\textwidth]{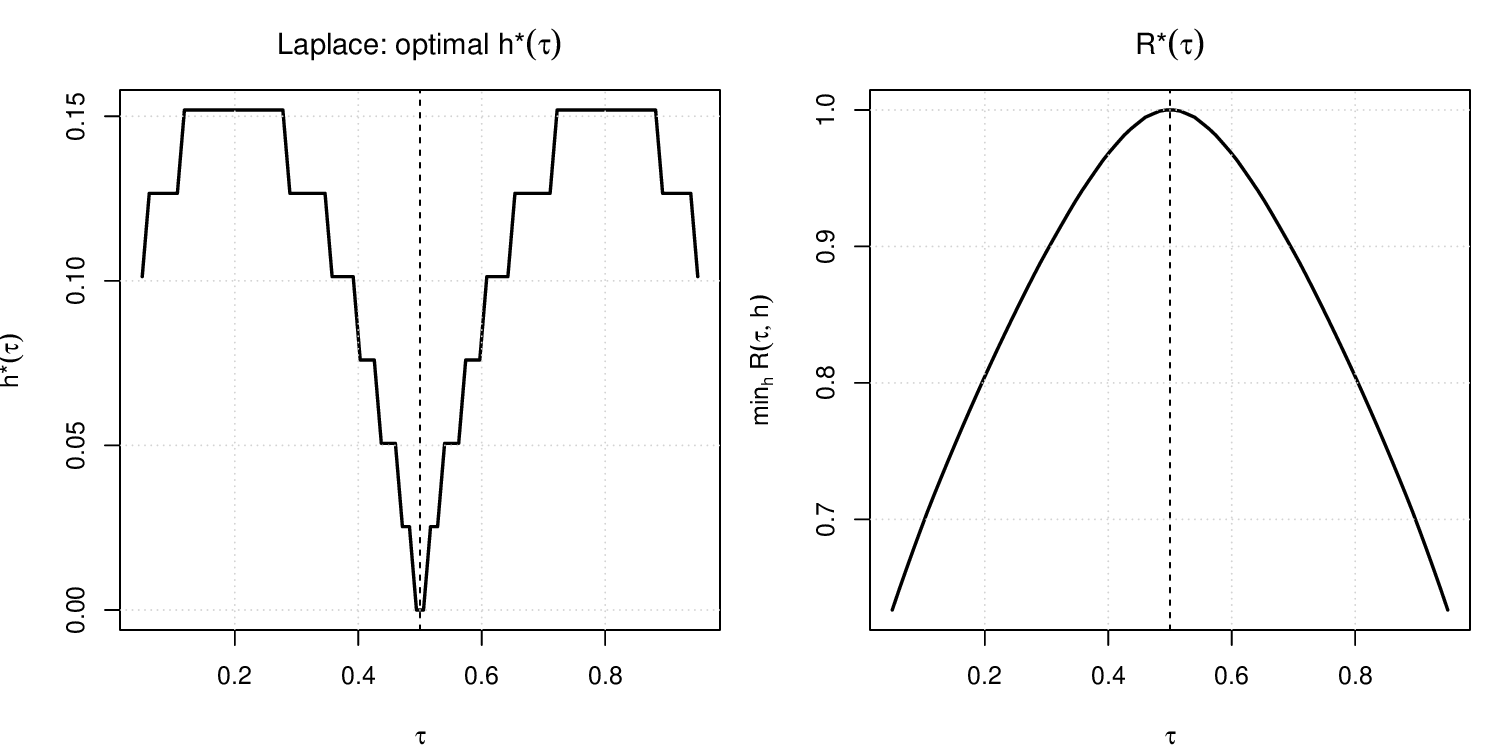}
	
	\caption{
		Optimal interpolation analysis.
		Top: Gaussian distribution.
		Bottom: Laplace distribution.
		For each distribution, the left panel displays the optimal interpolation
		strength
		$
		h_{\mathrm{eff}}^\star(\tau)
		=
		\arg\min_{h\ge0}R(\tau,h),
		$
		while the right panel displays the corresponding minimal efficiency ratio
		$
		R^\star(\tau)
		=
		\min_{h\ge0}R(\tau,h).
		$
	}
	\label{fig:optimal-h-comparison}
\end{figure}

For the Gaussian distribution, the minimum is attained at the largest
investigated value of $h$ for all quantile levels. Thus, within the
explored range, the efficiency ratio decreases monotonically with $h$,
indicating that stronger interpolation systematically improves asymptotic
efficiency. Moreover, the optimal ratio remains nearly constant over
$\tau$, taking values between approximately 0.638 and 0.647.
This corresponds to an asymptotic variance reduction of roughly
35\%-36\% relative to the empirical quantile estimator.
The symmetry of $R^\star(\tau)$ around $\tau=0.5$ reflects the
symmetry of the Gaussian distribution.
The Laplace distribution exhibits a markedly different behavior.
The optimal interpolation strength depends substantially on the quantile
level and satisfies the symmetry relation
\[
h_{\mathrm{eff}}^\star(\tau)=h_{\mathrm{eff}}^\star(1-\tau).
\]
The optimal value decreases as $\tau$ approaches the median and reaches
$
h_{\mathrm{eff}}^\star(0.5)=0.
$
Consequently, the empirical median is already asymptotically optimal,
and interpolation does not improve efficiency at the center of the
distribution. In contrast, the tails benefit considerably from
interpolation, with
$
h_{\mathrm{eff}}^\star(0.05)\approx 0.10,
$
$
h_{\mathrm{eff}}^\star(0.95)\approx 0.10.
$
yielding
$
R^\star(0.05)\approx 0.63,
$
$
R^\star(0.95)\approx 0.63.
$
This corresponds to an asymptotic variance reduction of approximately
37\% relative to the empirical quantile estimator.
As the quantile level approaches the center of the distribution,
the benefit of interpolation gradually vanishes and
$
R^\star(0.5)=1.
$
Hence no asymptotic gain can be achieved at the median.
These findings indicate that, for heavy-tailed distributions,
the usefulness of interpolation depends strongly on the quantile level,
with the greatest gains occurring in the tails.

Taken together, Figures~\ref{fig:contour-comparison}
and~\ref{fig:optimal-h-comparison}
show that the dependence of the efficiency ratio on $\tau$ is weak for
the Gaussian distribution but substantial for the Laplace distribution.
This contrast suggests that the optimal amount of interpolation is closely
linked to tail behavior: light-tailed distributions tend to benefit from
increasingly strong interpolation, whereas heavy-tailed distributions
admit a quantile-dependent optimal interpolation strength.

\subsubsection{Limiting interpolation paths}

The previous analysis considered the efficiency ratio
\[
R(\tau,h)
=
\frac{\sigma_Q^2(\tau,h)}
{\sigma_{\mathrm{emp}}^2(\tau_h)}
\]
over a two-dimensional grid of values of \(\tau\) and \(h\).
To further investigate the influence of the interpolation parameter,
we now examine the boundary cases
$
\tau=0
$
and
$
\tau=1.
$
These values should be interpreted as boundary interpolation paths
rather than ordinary quantile levels.
Indeed, for many continuous distributions the quantiles associated with
$\tau=0$ and $\tau=1$ are not finite.
Nevertheless, the interpolation equation
\[
F(q)+h(q-m)=\tau
\]
remains well defined and generates limiting trajectories that explore
the lower and upper tails of the distribution.
For fixed \(\tau\), the interpolation equation
$
F(q)+h(q-m)=\tau
$
defines a trajectory
$
h\mapsto q(\tau,h).
$
When \(\tau=0\), the resulting path traverses quantiles below the mean,
whereas \(\tau=1\) generates quantiles above the mean.
For each value of \(h\), we compare the asymptotic variance
$
\sigma_Q^2(\tau,h)
$
of the quadratic interpolation estimator with the asymptotic variance
\[
\sigma_{\mathrm{emp}}^2(F(q(\tau,h)))
=
\frac{F(q(\tau,h))
	\bigl(1-F(q(\tau,h))\bigr)}
{f(q(\tau,h))^2}
\]
of the empirical quantile estimator targeting the same population
quantile.
The resulting efficiency ratio is
\[
R(\tau,h)
=
\frac{\sigma_Q^2(\tau,h)}
{\sigma_{\mathrm{emp}}^2(F(q(\tau,h)))}.
\]
Figure~\ref{fig:extreme-paths}
shows the evolution of this ratio as a function of \(h\)
for the Gaussian and Laplace distributions.
\begin{figure}[!htbp]
	\centering
	\includegraphics[width=0.95\textwidth]{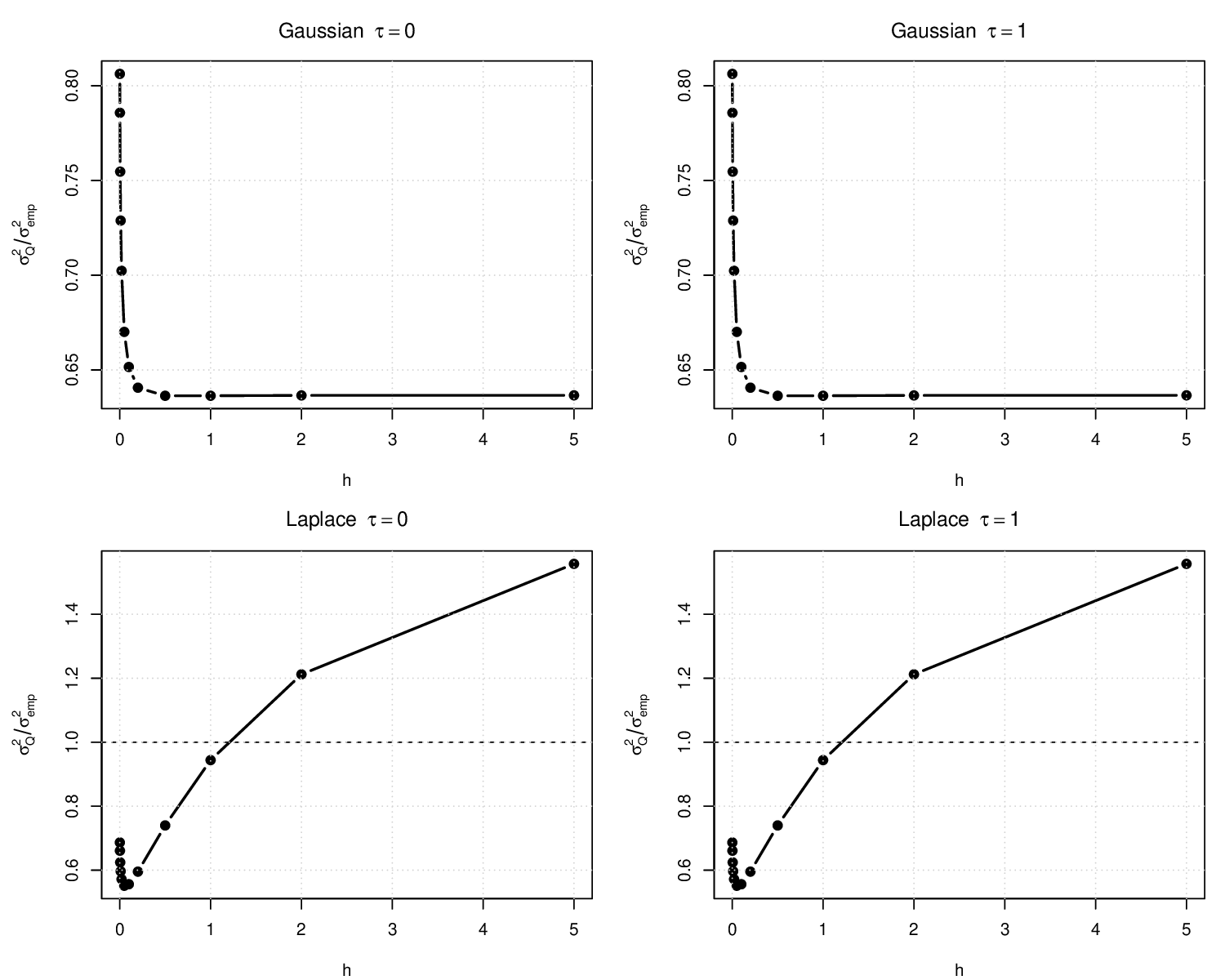}
	\caption{
		Efficiency ratio
		\(
		R(\tau,h)
		=
		\sigma_Q^2(\tau,h)
		/\sigma_{\mathrm{emp}}^2(F(q(\tau,h)))
		\)
		along the extreme interpolation paths corresponding to
		\(\tau=0\) and \(\tau=1\).
		The upper panels correspond to the Gaussian distribution and the lower
		panels to the Laplace distribution.
		The horizontal dashed line represents the benchmark level \(R=1\).
		Values below this line indicate that the quadratic interpolation
		estimator is asymptotically more efficient than the empirical quantile
		estimator targeting the same population quantile.
	}
	\label{fig:extreme-paths}
\end{figure}
The Gaussian distribution exhibits a remarkably stable behavior.
For both \(\tau=0\) and \(\tau=1\), the efficiency ratio remains
strictly below one over the entire range of values of \(h\)
investigated.
Moreover, the ratio decreases monotonically from approximately
\(0.81\) to \(0.64\) and then stabilizes.
Consequently, the quadratic interpolation estimator uniformly
outperforms the empirical quantile estimator along the entire path of
quantiles generated by the interpolation equation.
The limiting value
$
R(h)\approx 0.64
$
corresponds to an asymptotic variance reduction of approximately
\(36\%\).
The Laplace distribution displays a fundamentally different pattern.
The efficiency ratio initially decreases as \(h\) increases,
reaching a minimum value close to
$
R^\star \approx 0.55
$
for
$
h_{\mathrm{eff}}^\star \approx 0.05.
$
This corresponds to an asymptotic variance reduction of approximately
\(45\%\), which is substantially larger than in the Gaussian case.
However, beyond this optimal value the efficiency ratio increases
again.
Eventually it exceeds the threshold \(R=1\), indicating that excessive
interpolation becomes detrimental and that the empirical quantile
estimator becomes asymptotically preferable.
For both distributions, the curves corresponding to \(\tau=0\) and
\(\tau=1\) coincide up to numerical accuracy.
This reflects the symmetry relation
\[
R(0,h)=R(1,h),
\]
which follows from the symmetry of the Gaussian and Laplace
distributions around their means.

\noindent
These results provide additional evidence that the role of
interpolation is strongly linked to tail behavior.
For the Gaussian distribution, increasing interpolation strength
consistently improves efficiency.
In contrast, the Laplace distribution admits a genuine optimal
interpolation strength, beyond which additional interpolation leads to
a loss of efficiency.
This distinction complements the conclusions drawn from
Figures~\ref{fig:contour-comparison} and
\ref{fig:optimal-h-comparison} and further highlights the
distribution-dependent nature of the optimal interpolation strategy.

\subsubsection{Asymmetric distributions and tail-dependent efficiency}

The previous examples were based on the Gaussian and Laplace
distributions, both of which are symmetric around their means.
Consequently, the efficiency ratio satisfies the symmetry relation
$
R(0,h)=R(1,h),
$
and the behavior of the interpolation estimator is identical in the
lower and upper tails.
To investigate the effect of asymmetry, we consider the shifted
exponential distribution
\[
Y = Z-1,
\qquad
Z\sim \operatorname{Exp}(1),
\]
which has mean \(m=0\), variance \(\sigma_Y^2=1\), support
\([ -1,\infty )\), and density
\[
f(y)=e^{-(y+1)},
\qquad y\ge -1.
\]
Unlike the Gaussian and Laplace laws, this distribution is strongly
right-skewed.
As in the previous subsection, we examine the two extreme interpolation
paths generated by
$
F(q)+h(q-m)=\tau,
$
corresponding to \(\tau=0\) and \(\tau=1\).
For each value of \(h\), we compare the asymptotic variance of the
quadratic interpolation estimator,
$
\sigma_Q^2(\tau,h),
$
with the asymptotic variance of the empirical quantile estimator
targeting the same population quantile,
\[
\sigma_{\mathrm{emp}}^2(F(q))
=
\frac{F(q)(1-F(q))}
{f(q)^2}.
\]
The resulting efficiency ratio is
\[
R(\tau,h)
=
\frac{\sigma_Q^2(\tau,h)}
{\sigma_{\mathrm{emp}}^2(F(q))}.
\]

\begin{figure}[!htbp]
	\centering
	\includegraphics[width=0.9\textwidth]{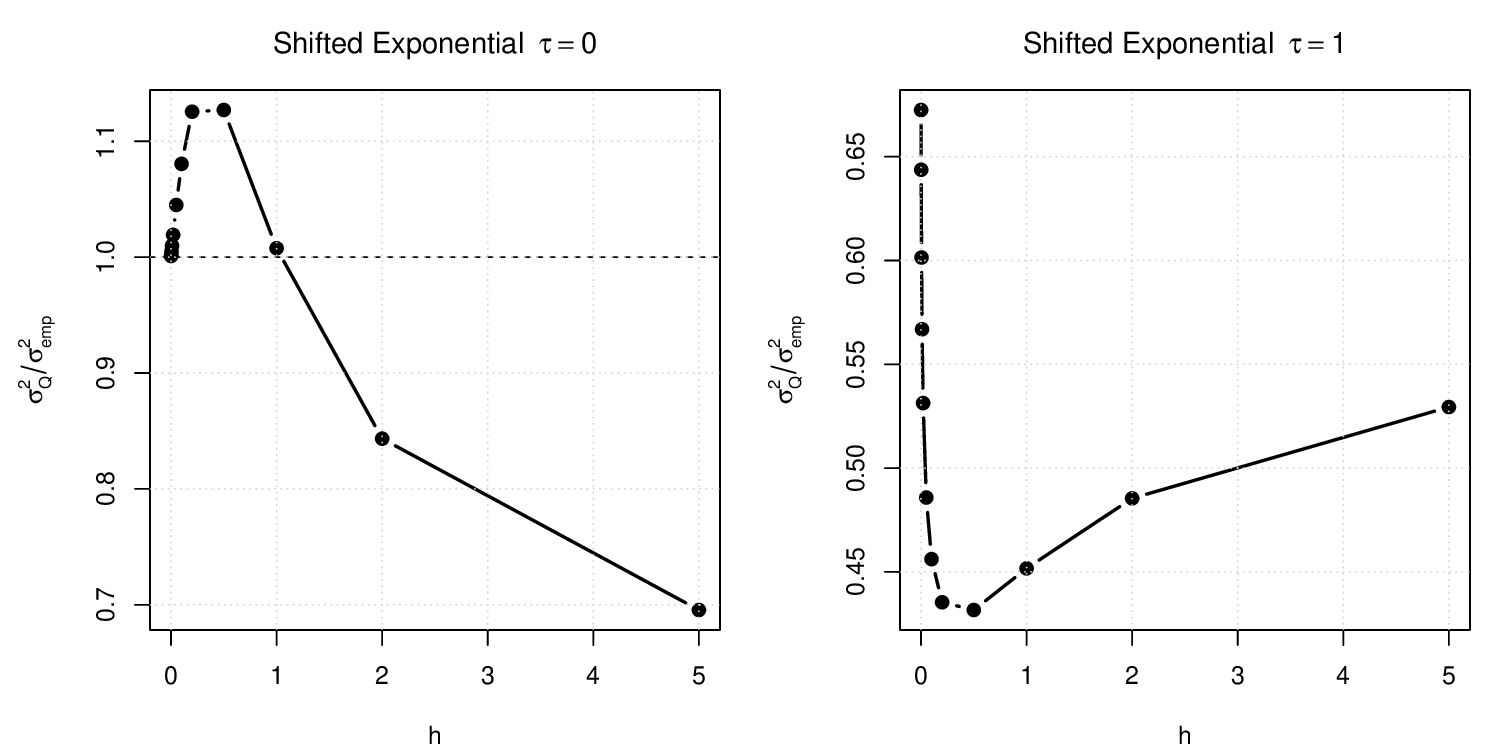}
	\caption{
		Efficiency ratio
		$
		R(\tau,h)
		=
		\frac{\sigma_Q^2(\tau,h)}
		{\sigma_{\mathrm{emp}}^2(F(q))}
		$
		along the extreme interpolation paths for the shifted exponential
		distribution.
		The left panel corresponds to \(\tau=0\), while the right panel
		corresponds to \(\tau=1\).
		The dashed horizontal line indicates the benchmark level \(R=1\).
	}
	\label{fig:shifted-exp-extreme}
\end{figure}

\noindent
Figure~\ref{fig:shifted-exp-extreme} reveals a striking asymmetry
between the lower and upper tails.
For \(\tau=0\), the efficiency ratio initially exceeds one,
reaching values close to
$
R \approx 1.13.
$
Thus, for a broad range of moderate interpolation strengths,
the empirical quantile estimator is asymptotically more efficient than
the interpolated estimator.
Only for sufficiently large values of \(h\) does the efficiency ratio
drop below one, eventually reaching values around
$
R \approx 0.70.
$
Hence the gain produced by interpolation in the lower tail is
relatively modest and appears only for strong interpolation.
The behavior is completely different for \(\tau=1\).
In this case the efficiency ratio remains strictly below one over the
entire range of values of \(h\) considered.
The minimum value is attained near
$
h_{\mathrm{eff}}^\star \approx 0.5,
$
for which
$
R^\star \approx 0.43.
$
This corresponds to an asymptotic variance reduction of approximately
\(57\%\) relative to the empirical quantile estimator.
Therefore, interpolation is uniformly beneficial in the upper tail and
provides substantially larger gains than in the symmetric examples.
The asymmetry observed in Figure~\ref{fig:shifted-exp-extreme}
can be interpreted using the variance representation of
Proposition~\ref{prop:variance-structure}.
The denominator of
\[
\sigma_Q^2(\tau,h)
=
\frac{
	F(q)(1-F(q))
	+
	h^2\sigma_Y^2
	-
	2hT(q)
}
{(f(q)+h)^2}
\]
contains the local density \(f(q)\).
Near the left endpoint of the support,
the density remains relatively large,
so the empirical quantile estimator is already highly efficient and
there is little room for improvement through interpolation.
In contrast, the upper tail corresponds to values of \(q\) for which
\(f(q)\) becomes very small.
The interpolation term \(h\) then acts as a regularization mechanism,
preventing the denominator from becoming excessively small and
substantially reducing the asymptotic variance.

These results demonstrate that the effectiveness of interpolation is
influenced not only by tail heaviness but also by distributional
asymmetry.
For asymmetric distributions, the gain achieved by interpolation may
differ dramatically between the lower and upper tails.
The shifted exponential distribution therefore provides a concrete
example where interpolation is highly beneficial in one tail while
offering only limited improvement in the other.

\section{Huber interpolation}
\label{sec:two}
\subsection{Definition and first-order condition}

We also consider a robust extension based on the Huber loss
\[
\rho_H(u)
=
\begin{cases}
	\frac12u^2,
	&
	|u|\le k,
	\\[2mm]
	k|u|-\frac12k^2,
	&
	|u|>k,
\end{cases}
\]
where
$
k>0
$
is a tuning parameter.
The corresponding population Huber interpolation estimator is defined by
\[
q_H(\tau,h)
=
\arg\min_{q\in\mathbb R}
\left\{
\mathbb E[\rho_\tau(Y-q)]
+
h\,
\mathbb E[\rho_H(Y-q)]
\right\},
\]
with empirical counterpart
\[
\hat q_H(\tau,h)
=
\arg\min_{q\in\mathbb R}
\left\{
\sum_{i=1}^n
\rho_\tau(y_i-q)
+
h
\sum_{i=1}^n
\rho_H(y_i-q)
\right\}.
\]
The associated first-order condition is
\[
F(q)
-
h\,\mathbb E[\psi_H(Y-q)]
=
\tau,
\]
where
$
\psi_H=\rho_H'
$
denotes the Huber score function.
Thus, the Huber interpolation also defines a probability-space deformation mechanism, but with bounded influence contributions replacing the quadratic correction term.

\begin{proposition}
	\label{prop:huber-existence}
	
	Assume that $F$ admits a continuous density $f$, that
	$\mathbb E|Y|<\infty$, and that
	$
	f(q)>0
	$
	for every $q$ in the support of $Y$.
	For every $\tau\in(0,1)$ and every $h\ge0$,
	the equation
	$
	F(q)-h\,\mathbb E[\psi_H(Y-q)]
	=
	\tau
	$
	admits a unique solution
	$q_H(\tau,h)\in\mathbb R$.
	
\end{proposition}

\begin{proof}
	
	Define
	\[
	G_H(q)
	=
	F(q)-h\,\mathbb E[\psi_H(Y-q)].
	\]
	Since $\psi_H$ is bounded and continuous,
	the map
	$
	q\mapsto \mathbb E[\psi_H(Y-q)]
	$
	is continuous by dominated convergence.
	Hence $G_H$ is continuous.
	Furthermore,
	\[
	\frac{d}{dq}
	\mathbb E[\psi_H(Y-q)]
	=
	-\mathbb P(|Y-q|\le k),
	\]
	By assumption,
	$
	G_H'(q)
	=
	f(q)
	+
	h\,\mathbb P(|Y-q|\le k)
	>
	0,
	$
	and therefore $G_H$ is strictly increasing.
	Since $\psi_H$ is bounded,
	\[
	\lim_{q\to-\infty}
	\mathbb E[\psi_H(Y-q)]
	=
	k,
	\qquad
	\lim_{q\to+\infty}
	\mathbb E[\psi_H(Y-q)]
	=
	-k.
	\]
	Consequently,
	\[
	\lim_{q\to-\infty}
	G_H(q)
	=
	-hk,
	\qquad
	\lim_{q\to+\infty}
	G_H(q)
	=
	1+hk.
	\]
	Because
	$
	-hk < \tau < 1+hk,
	$
	the intermediate value theorem implies existence of a solution.
	Strict monotonicity yields uniqueness.
\end{proof}

\subsection{Asymptotic normality}

The Huber interpolation estimator also belongs to the class of
one-dimensional M-estimators.
Define
\[
\Psi_H(y,q)
=
\tau-\mathbf 1_{\{y<q\}}
+
h\,\psi_H(y-q),
\]
where
\[
\psi_H(u)
=
\begin{cases}
	u,
	&
	|u|\le k,
	\\[2mm]
	k\,\mathrm{sign}(u),
	&
	|u|>k,
\end{cases}
\]
denotes the Huber score function.
The population estimator \(q_H(\tau,h)\) is characterized by
$
\mathbb E[\Psi_H(Y,q_H(\tau,h))]=0,
$
which is equivalent to the first-order condition
\[
F(q_H(\tau,h))
-
h\,\mathbb E\!\left[\psi_H(Y-q_H(\tau,h))\right]
=
\tau.
\]
The following theorem establishes the asymptotic normality of the
Huber interpolation estimator.

\begin{theorem}
	\label{thm:huber-clt}
	Assume that \(Y_1,\ldots,Y_n\) are i.i.d. observations from a distribution
	with cumulative distribution function \(F\) and density \(f\),
	let $h\ge0$ and $k>0$ be fixed constants independent of $n$.
	Under the assumptions of Proposition~\ref{prop:huber-existence},
	let \(q_H(\tau,h)\) denote the unique solution of
	$
	F(q)
	-
	h\,\mathbb E[\psi_H(Y-q)]
	=
	\tau.
	$
	Then
	\[
	\sqrt n
	\Bigl(
	\hat q_H(\tau,h)-q_H(\tau,h)
	\Bigr)
	\xrightarrow{d}
	\mathcal N
	\!\left(
	0,
	\sigma_H^2(\tau,h)
	\right),
	\]
	where
	\[
	\sigma_H^2(\tau,h)
	=
	\frac{B_H(\tau,h)}
	{A_H(\tau,h)^2},
	\]
	with
	\[
	A_H(\tau,h)
	=
	f(q_H(\tau,h))
	+
	h\,\mathbb P
	\!\left(
	|Y-q_H(\tau,h)|\le k
	\right),
	\]
	and
	\[
	B_H(\tau,h)
	=
	\mathbb E
	\Bigl[
	\bigl(
	\tau
	-
	\mathbf 1_{\{Y<q_H(\tau,h)\}}
	+
	h\,\psi_H(Y-q_H(\tau,h))
	\bigr)^2
	\Bigr].
	\]
\end{theorem}

\begin{proof}
	
	The estimator $\hat q_H(\tau,h)$ is the unique solution of the
	empirical estimating equation
	\[
	\frac1n
	\sum_{i=1}^n
	\Bigl(
	\tau
	-
	\mathbf 1_{\{Y_i<q\}}
	+
	h\,\psi_H(Y_i-q)
	\Bigr)
	=
	0.
	\]
	Since the score contains the discontinuous indicator function
	$\mathbf 1_{\{Y_i<q\}}$, the standard smooth M-estimation theory
	does not apply directly.
	A rigorous derivation based on empirical-process arguments is
	given in Appendix~\ref{app:proof-huber-clt}.
	The appendix establishes the asymptotic linear representation
	\[
	\sqrt n
	\Bigl(
	\hat q_H(\tau,h)-q_H(\tau,h)
	\Bigr)
	=
	\frac1{A_H(\tau,h)}
	\frac1{\sqrt n}
	\sum_{i=1}^n
	\Psi_H(Y_i,q_H(\tau,h))
	+
	o_p(1),
	\]
	where
	\[
	A_H(\tau,h)
	=
	f(q_H(\tau,h))
	+
	h\,\mathbb P
	\!\left(
	|Y-q_H(\tau,h)|\le k
	\right).
	\]
	Applying the central limit theorem to the i.i.d. summands
	$\Psi_H(Y_i,q_H(\tau,h))$ yields
	\[
	\sqrt n
	\Bigl(
	\hat q_H(\tau,h)-q_H(\tau,h)
	\Bigr)
	\xrightarrow{d}
	\mathcal N
	\!\left(
	0,
	\frac{B_H(\tau,h)}
	{A_H(\tau,h)^2}
	\right),
	\]
	where
	\[
	B_H(\tau,h)
	=
	\mathbb E
	\Bigl[
	\Psi_H(Y,q_H(\tau,h))^2
	\Bigr].
	\]
\end{proof}

\begin{remark}
	When \(h=0\), the Huber interpolation estimator reduces to the classical
	empirical quantile estimator. In this case,
	$
	A_H(\tau,0)=f(q_\tau),
	$
	and
	$
	B_H(\tau,0)=\tau(1-\tau),
	$
	and therefore
	\[
	\sigma_H^2(\tau,0)
	=
	\frac{\tau(1-\tau)}
	{f(q_\tau)^2},
	\]
	which is the well-known asymptotic variance of the empirical quantile.
\end{remark}

\subsection{Probability-space interpretation}

As in the quadratic interpolation case, the Huber estimator may be interpreted as inducing a deformation of the underlying probability level associated with the target quantile.
Let
$
q_H(\tau,h)
$
denote the population Huber interpolation estimator, defined by
\[
F(q_H)
-
h\,\mathbb E[\psi_H(Y-q_H)]
=
\tau.
\]
Introducing
\[
p_H(\tau,h)
=
F(q_H(\tau,h)),
\]
the defining equation may be rewritten as
\[
p_H(\tau,h)
=
\tau
+
h\,\mathbb E
\!\left[
\psi_H\!\bigl(Y-q_H(\tau,h)\bigr)
\right].
\]
Here and throughout this section,
\[
F^{-1}(u)
=
\inf\{x\in\mathbb R:\,F(x)\ge u\}
\]
denotes the generalized inverse of $F$.
Consequently,
\[
q_H(\tau,h)
=
F^{-1}\!\bigl(p_H(\tau,h)\bigr),
\]
showing that the Huber interpolation also acts through a deformation of the probability level associated with the quantile.
Unlike the quadratic interpolation, however, the deformation is implicit because
\(q_H(\tau,h)\) appears inside the expectation.
The probability level therefore satisfies the nonlinear fixed-point relation
\[
p_H(\tau,h)
=
\tau
+
h\,\mathbb E
\!\left[
\psi_H
\!\left(
Y-F^{-1}(p_H(\tau,h))
\right)
\right].
\]
Hence the Huber interpolation does not correspond to a simple affine shift of the quantile level. Instead, the deformation depends on the entire distribution through the bounded score function \(\psi_H\).
For every fixed regularization parameter $h\ge0$,
the boundedness of the Huber score implies that
$
|\psi_H(u)|\le k,
$
and therefore
$
|p_H(\tau,h)-\tau|
\le hk.
$
In contrast with the quadratic interpolation, where the deformation magnitude grows proportionally to the residual size, the Huber deformation remains uniformly bounded. Large observations therefore cannot induce arbitrarily large shifts of the effective probability level.
For small values of \(h\), the deformation admits the first-order approximation
\[
p_H(\tau,h)
=
\tau
+
h\,\mathbb E
\!\left[
\psi_H(Y-q_\tau)
\right]
+
o(h),
\]
where \(q_\tau=F^{-1}(\tau)\).
Thus the Huber interpolation may be viewed as a robust probability-space deformation of the empirical quantile, analogous to the quadratic interpolation but with bounded influence contributions replacing the linear correction term.

\section{Bisquare interpolation}
\label{sec:three}
\subsection{Definition and first-order condition}

We now consider an interpolation between the quantile loss and Tukey's
bisquare loss~\citep{Beaton1974}.
Let
\[
\rho_B(u)
=
\begin{cases}
	\dfrac{k^2}{6}
	\left[
	1-
	\left(
	1-\dfrac{u^2}{k^2}
	\right)^3
	\right],
	&
	|u|\le k,
	\\[2ex]
	\dfrac{k^2}{6},
	&
	|u|>k,
\end{cases}
\]
where \(k>0\) is a tuning parameter.
For a fixed quantile level
$
\tau\in(0,1),
$
and interpolation parameter
$
h\ge0,
$
define the population criterion
\[
M_h(q)
=
\mathbb E[\rho_\tau(Y-q)]
+
h\,\mathbb E[\rho_B(Y-q)].
\]
Any global minimizer of \(M_h\), whenever it exists, satisfies the
first-order optimality condition established below.

\begin{proposition}
	\label{prop:bisquare-population}
	
	Assume that the distribution function \(F\) of \(Y\) is continuous.
	Then every minimizer \(q_{\mathrm{B}}(\tau,h)\) of \(M_h\) satisfies
	\[
	\tau-F(q_{\mathrm{B}}(\tau,h))
	+
	h\,B(q_{\mathrm{B}}(\tau,h))
	=
	0,
	\]
	where
	\[
	B(q)
	=
	\mathbb E
	\!\left[
	(Y-q)
	\left(
	1-\frac{(Y-q)^2}{k^2}
	\right)^2
	\mathbf 1_{\{|Y-q|\le k\}}
	\right].
	\]
\end{proposition}

\begin{proof}
	
	The derivative of the quantile loss is
	$
	\psi_\tau(u)
	=
	\tau-\mathbf 1_{\{u<0\}}.
	$
	The derivative of Tukey's bisquare loss is
	\[
	\psi_B(u)
	=
	u
	\left(
	1-\frac{u^2}{k^2}
	\right)^2
	\mathbf 1_{\{|u|\le k\}}.
	\]
	Differentiating \(M_h\) yields
	\[
	M_h'(q)
	=
	-
	\mathbb E[\psi_\tau(Y-q)]
	-
	h\,
	\mathbb E[\psi_B(Y-q)].
	\]
	The first-order optimality condition is therefore
	\[
	\mathbb E[\psi_\tau(Y-q)]
	+
	h
	\mathbb E[\psi_B(Y-q)]
	=
	0.
	\]
	Since
	$
	\mathbb E[\psi_\tau(Y-q)]
	=
	\tau-F(q),
	$
	we obtain
	\[
	\tau-F(q)
	+
	h
	\mathbb E
	\!\left[
	(Y-q)
	\left(
	1-\frac{(Y-q)^2}{k^2}
	\right)^2
	\mathbf 1_{\{|Y-q|\le k\}}
	\right]
	=
	0.
	\]
	The result follows.
\end{proof}

\subsection{Existence and uniqueness}

The following result establishes the existence of a solution to the
estimating equation derived in Proposition~\ref{prop:bisquare-population}.

\begin{proposition}
	\label{prop:bisquare-existence}
	
	Assume that \(F\) is continuous.
	Then the estimating equation
	$
	\tau-F(q)+hB(q)=0
	$
	admits at least one solution.
	If
	$
	G_h(q)
	=
	\tau-F(q)+hB(q)
	$
	is strictly decreasing, then the solution is unique.
\end{proposition}

\begin{proof}
	
	Since
	$
	\lim_{q\to-\infty}F(q)=0,
	$
	$
	\lim_{q\to+\infty}F(q)=1,
	$
	and
	$
	\lim_{|q|\to\infty}B(q)=0,
	$
	it follows that
	$
	\lim_{q\to-\infty}G_h(q)
	=
	\tau
	>
	0,
	$
	and
	$
	\lim_{q\to+\infty}G_h(q)
	=
	\tau-1
	<
	0.
	$
	The existence of a solution follows from the intermediate value theorem.
	If \(G_h\) is strictly decreasing, the solution is necessarily unique.
\end{proof}

\begin{remark}
	\label{rem:bisquare-selection}
	
	Throughout the remainder of this section, we let
	\(q_{\mathrm B}(\tau,h)\)
	denote the selected regular solution of the estimating equation
	$
	\tau-F(q)+hB(q)=0,
	$
	around which the asymptotic analysis is carried out.
	Unlike the quadratic and Huber losses, Tukey's bisquare loss is
	non-convex. Consequently, a solution of the estimating equation
	$
	\mathbb E[\Psi_h(Y-q)]=0
	$
	need not be a global minimizer of the population criterion \(M_h\).
	We therefore assume that the estimating equation admits a unique regular
	solution \(q_{\mathrm B}(\tau,h)\), locally identified by
	$
	A_B(\tau,h)\neq0,
	$
	and that the empirical estimator denotes the corresponding solution of
	the empirical estimating equation converging in probability to this
	population solution. Under additional conditions ensuring that this
	solution is also the unique global minimizer of \(M_h\), the two notions
	coincide.
\end{remark}

\subsection{Asymptotic normality}

\noindent
Let
$
Y_1,\ldots,Y_n
$
be an i.i.d. sample.
The empirical interpolated bisquare criterion is
\[
M_{n,h}(q)
=
\sum_{i=1}^{n}
\rho_\tau(Y_i-q)
+
h
\sum_{i=1}^{n}
\rho_B(Y_i-q).
\]
An empirical interpolated bisquare estimator may be defined as any global minimizer of \(M_{n,h}\). Throughout the remainder of this section, however, we denote by \(\hat q_{\mathrm B}(\tau,h)\) the selected empirical solution of the estimating equation converging in probability to the selected population solution \(q_{\mathrm B}(\tau,h)\).
Define
\[
\Psi_h(u)
=
\psi_\tau(u)
+
h\psi_B(u).
\]
Then
\[
\mathbb E[\Psi_h(Y-q_{\mathrm B}(\tau,h))]
=
0,
\]
and
\(\hat q_{\mathrm B}(\tau,h)\)
satisfies the corresponding empirical estimating equation.

\begin{theorem}
	\label{thm:bisquare-clt}
	
	Assume that the regularization parameter \(h\) and the bisquare tuning
	parameter \(k\) are fixed.
	Let \(q_{\mathrm B}(\tau,h)\) denote the selected solution of the
	estimating equation
	$
	\mathbb E[\Psi_h(Y-q)]=0,
	$
	and suppose that this solution is locally unique, that
	$
	\mathbb E[\Psi_h(Y-q_{\mathrm B}(\tau,h))^2]<\infty,
	$
	and that
	\[
	A_B(\tau,h)
	=
	-
	\mathbb E[\Psi_h'(Y-q_{\mathrm B}(\tau,h))]
	\neq0.
	\]
	Let
	$\hat q_{\mathrm B}(\tau,h)$
	denote the selected empirical solution introduced above.
	Then
	\[
	\sqrt n
	(\hat q_{\mathrm{B}}(\tau,h)-q_{\mathrm{B}}(\tau,h))
	\overset d\longrightarrow
	\mathcal N(0,\sigma_{\mathrm{B}}^2(\tau,h)),
	\]
	where
	\[
	\sigma_{\mathrm{B}}^2(\tau,h)
	=
	\frac{B_B(\tau,h)}{A_B^2(\tau,h)},
	\]
	with
	\[
	B_B(\tau,h)
	=
	\mathbb E
	\left[
	\Psi_h(Y-q_{\mathrm{B}}(\tau,h))^2
	\right].
	\]
	Equivalently,
	\[
	\sigma_{\mathrm{B}}^2(\tau,h)
	=
	\frac{
		\mathbb E
		\left[
		\bigl(
		\psi_\tau(Y-q_{\mathrm{B}}(\tau,h))
		+
		h\psi_B(Y-q_{\mathrm{B}}(\tau,h))
		\bigr)^2
		\right]
	}
	{
		\left(
		f(q_{\mathrm{B}}(\tau,h))
		+
		h\,\mathbb E[\psi_B'(Y-q_{\mathrm{B}}(\tau,h))]
		\right)^2
	}.
	\]
\end{theorem}

\begin{proof}
	By construction,
	\(\hat q_{\mathrm B}(\tau,h)\)
	denotes the selected empirical solution of the estimating equation
	converging in probability to
	\(q_{\mathrm B}(\tau,h)\).
	Therefore,
	\[
	\frac1n
	\sum_{i=1}^{n}
	\left(
	\psi_\tau(Y_i-\hat q_{\mathrm B}(\tau,h))
	+
	h\,\psi_B(Y_i-\hat q_{\mathrm B}(\tau,h))
	\right)
	=
	0.
	\]
	Since the score contains the discontinuous quantile component
	$\psi_\tau$, the standard smooth M-estimation theory does not
	apply directly.
	A rigorous derivation based on empirical-process methods is given in
	Appendix~\ref{app:proof-bisquare-clt}.
	The appendix establishes the asymptotic linear representation
	\[
	\sqrt n
	\Bigl(
	\hat q_{\mathrm B}(\tau,h)
	-
	q_{\mathrm B}(\tau,h)
	\Bigr)
	=
	\frac1{A_B(\tau,h)}
	\frac1{\sqrt n}
	\sum_{i=1}^{n}
	\Psi_h(Y_i-q_{\mathrm B}(\tau,h))
	+
	o_p(1),
	\]
	where
	\[
	A_B(\tau,h)
	=
	f(q_{\mathrm B}(\tau,h))
	+
	h\,
	\mathbb E
	\!\left[
	\psi_B'(Y-q_{\mathrm B}(\tau,h))
	\right].
	\]
	Applying the central limit theorem to the i.i.d. summands
	$\Psi_h(Y_i-q_{\mathrm B}(\tau,h))$
	yields
	\[
	\sqrt n
	\Bigl(
	\hat q_{\mathrm B}(\tau,h)
	-
	q_{\mathrm B}(\tau,h)
	\Bigr)
	\xrightarrow d
	\mathcal N
	\!\left(
	0,
	\frac{B_B(\tau,h)}
	{A_B(\tau,h)^2}
	\right),
	\]
	where
	\[
	B_B(\tau,h)
	=
	\mathbb E
	\!\left[
	\Psi_h(Y-q_{\mathrm B}(\tau,h))^2
	\right].
	\]
\end{proof}

\subsection{Probability-space interpretation}

\noindent
Let
$
p_B(h)=F(q_{\mathrm{B}}(\tau,h)),
$
where \(q_{\mathrm{B}}(\tau,h)\) denotes the selected solution of the estimating equation
\[
\tau-F(q_{\mathrm{B}}(\tau,h))
+h\,\mathbb E[\psi_B(Y-q_{\mathrm{B}}(\tau,h))]
=0.
\]
Then
\[
p_B(h)
=
\tau
+
h\,\mathbb E[\psi_B(Y-q_{\mathrm{B}}(\tau,h))].
\]
Provided that the selected solution
\(q_{\mathrm B}(\tau,h)\)
depends continuously on the interpolation parameter \(h\), the latter
induces a continuous deformation of the target quantile level.
Unlike the quadratic interpolation, however, no explicit closed-form
relation between \(p_B(h)\) and \(h\) is available in general because
\(q_{\mathrm B}(\tau,h)\) appears inside the expectation and the
non-convexity of the bisquare criterion may lead to multiple solution
branches.
For a prescribed target probability level
\(
\tau'\in(0,1),
\)
one may search numerically for interpolation parameters \(h\) satisfying
$
F(q_{\mathrm B}(\tau,h))
=
\tau'.
$
Unlike the quadratic interpolation, the existence and uniqueness of such
an interpolation parameter are not guaranteed in general, since the map
\(h\mapsto q_{\mathrm B}(\tau,h)\) need not be monotone or globally
continuous.
In practice, one may estimate
$
p_B(h)
=
F(q_{\mathrm B}(\tau,h))
$
empirically through
$
\hat p_B(h)
=
F_n(\hat q_{\mathrm B}(\tau,h)),
$
and search numerically for values of \(h\) satisfying
\(
\hat p_B(h)\approx\tau',
\)
whenever such values exist.

\section{Regression extensions}
\label{sec:reg}
\subsection{From location to regression}

The location estimators studied in the preceding sections illustrate the
main ideas underlying the proposed interpolation framework. The
regularization parameter generates a continuous path between the
empirical quantile and a central location estimator, while the
asymptotic analysis identifies conditions under which interpolation can
improve estimation efficiency. These structural properties are not
specific to the scalar setting. They arise directly from the penalized
objective function and extend naturally to the linear regression model
\[
Y_i = X_i^\top \beta + \varepsilon_i,
\qquad i = 1,\dots,n,
\]
where the scalar parameter $q$ is replaced by the linear predictor
$X_i^\top\beta$.
In the regression setting, the interpolation path operates in the
\(p\)-dimensional parameter space, and the efficiency question takes on
a richer form. Rather than comparing scalar asymptotic variances, one
must compare asymptotic covariance matrices, with the natural benchmark
being the Cramér--Rao lower bound under the asymmetric Laplace model,
which classical quantile regression is known to attain asymptotically.
The central question is therefore whether the interpolated regression
estimators can improve the finite-sample performance of the classical
quantile regression estimator, and whether the structural advantage of
the Huber interpolation over the quadratic interpolation, already
observed in the location setting, persists in the regression context.

The remainder of the paper develops this program.
Sections~\ref{sec:model}--\ref{sec:quantreg} recall the likelihood
structure, score function, and Cramér--Rao lower bound under the
asymmetric Laplace model.
Section~\ref{sec:interpolated-regression} introduces the interpolated
regression estimators and establishes their asymptotic distribution.
The finite-sample performance of the proposed estimators is then
investigated through a Monte Carlo study, followed by a discussion of
the numerical results.

\subsection{Model and likelihood representation}
\label{sec:model}

We consider the linear regression model
\[
Y_i = X_i^\top \beta + \varepsilon_i, \qquad i=1,\dots,n,
\]
where $X_i \in \mathbb{R}^p$ are i.i.d. covariates with finite second moments and non-singular covariance matrix $\mathbb{E}[X X^\top]$, and $\varepsilon_i$ are i.i.d. error terms independent of $X_i$.
We assume that the error distribution admits a density $f_\varepsilon$ which is continuous and strictly positive in a neighborhood of zero. In the correctly specified case, we further assume that $\varepsilon_i$ follows an asymmetric Laplace distribution (ALD) with quantile level $\tau \in (0,1)$ and scale parameter $\sigma>0$. Its density can be written as
\[
f_\varepsilon(u)
=
\frac{\tau(1-\tau)}{\sigma}
\exp\left(-\frac{\rho_\tau(u)}{\sigma}\right),
\]
where $\rho_\tau(u) = u(\tau - \mathbf{1}_{\{u<0\}})$ denotes the check loss function.
Under this parametrization, the log-likelihood for $\beta$, up to an additive constant, is
\[
\ell_n(\beta)
=
-\frac{1}{\sigma} \sum_{i=1}^n \rho_\tau\big(Y_i - X_i^\top \beta\big),
\]
so that maximizing the likelihood is equivalent to solving the quantile regression problem.

\subsection{Score function and Fisher information}
\label{sec:score}

Let $u_i(\beta) = Y_i - X_i^\top \beta$. The score function is given by
\[
\nabla_\beta \ell_n(\beta)
=
\frac{1}{\sigma}
\sum_{i=1}^n
\psi_\tau\big(u_i(\beta)\big)\, X_i,
\]
where $\psi_\tau(u) = \tau - \mathbf{1}_{\{u<0\}}$ is the subgradient of $\rho_\tau$.
Under standard regularity conditions, the Fisher information matrix is defined as
\[
I(\beta)
=
\mathbb{E}\left[
-\nabla^2_\beta \ell_n(\beta)
\right]
=
n \cdot \frac{f_\varepsilon(0)}{\sigma^2}
\, \mathbb{E}[X X^\top].
\]
For the asymmetric Laplace distribution, we have
\[
f_\varepsilon(0) = \frac{\tau(1-\tau)}{\sigma},
\]
which will be used below to derive the corresponding Cramér--Rao lower bound.

\subsection{Cramér--Rao lower bound}
\label{sec:crlb}

The Cramér--Rao lower bound (CRLB) states that for any unbiased estimator $\hat{\beta}$,
\[
\mathrm{Var}(\hat{\beta})
\succeq
I(\beta)^{-1}.
\]
	Under the asymmetric Laplace model, this yields
	\[
	\operatorname{Var}(\hat\beta)
	\succeq
	\frac{\sigma^2}
	{n\,\tau(1-\tau)}
	\bigl(\mathbb E[XX^\top]\bigr)^{-1}.
	\]
	Taking traces gives the scalar bound
	\[
	\operatorname{tr}\operatorname{Var}(\hat\beta)
	\ge
	\frac{\sigma^2}
	{n\,\tau(1-\tau)}
	\operatorname{tr}
	\!\left(
	(\mathbb E[XX^\top])^{-1}
	\right).
	\]
	This provides the theoretical efficiency benchmark used in our numerical experiments.

\subsection{Quantile regression estimator}
\label{sec:quantreg}

The quantile regression estimator is defined as
\[
\hat{\beta}_{\tau}
=
\arg\min_{\beta}
\sum_{i=1}^n
\rho_\tau\big(Y_i - X_i^\top \beta\big).
\]
Under the ALD model, this estimator coincides with the maximum likelihood estimator. Under standard assumptions (see, e.g., asymptotic theory for M-estimators), it satisfies
\[
\sqrt{n}(\hat{\beta}_\tau - \beta)
\xrightarrow{d}
\mathcal{N}\left(
0,\,
\frac{\tau(1-\tau)}{f_\varepsilon(0)^2}
\big(\mathbb{E}[X X^\top]\big)^{-1}
\right).
\]
Since $f_\varepsilon(0)=\tau(1-\tau)/\sigma,$
the asymptotic covariance simplifies to
\[
\operatorname{Var}(\hat\beta_\tau)
=
\frac{\sigma^2}
{n\,\tau(1-\tau)}
\bigl(\mathbb E[XX^\top]\bigr)^{-1},
\]
showing that the quantile regression estimator is
\emph{asymptotically efficient and attains the Cramér--Rao lower bound}.

\subsection{Quadratic and Huber interpolated estimators}
\label{sec:interpolated-regression}

We consider two families of interpolated estimators obtained by combining the quantile loss with either a quadratic or a Huber regularization.
For a fixed quantile level \(\tau\), let
$
\rho_\tau(u)
=
u\bigl(\tau-\mathbf 1_{\{u<0\}}\bigr)
$
denote the quantile check loss.
The quadratic interpolation is defined by
\[
\hat\beta_\alpha^{(Q)}
=
\arg\min_\beta
\sum_{i=1}^n
\left[
(1-\alpha)\rho_\tau(u_i(\beta))
+
\alpha\,u_i(\beta)^2
\right],
\]
while the Huber interpolation is defined by
\[
\hat\beta_\alpha^{(H)}
=
\arg\min_\beta
\sum_{i=1}^n
\left[
(1-\alpha)\rho_\tau(u_i(\beta))
+
\alpha\,\rho_H(u_i(\beta))
\right],
\]
where \(\rho_H\) denotes the Huber loss and \(\alpha\in[0,1]\) is an interpolation parameter.
For \(\alpha<1\), multiplying the objective by the positive constant \(1/(1-\alpha)\) does not change the minimizer. Consequently, the quadratic interpolation may equivalently be written in the additive form
\[
\rho_\tau(u)+h\,u^2,
\qquad
h=\frac{\alpha}{1-\alpha},
\]
and the Huber interpolation may equivalently be written as
\[
\rho_\tau(u)+h\,\rho_H(u),
\qquad
h=\frac{\alpha}{1-\alpha}.
\]
The convex-combination parametrization is retained throughout the paper because it provides a direct interpretation of the interpolation path on the unit interval:
$
\alpha=0
\Longrightarrow\quad
$
pure quantile regression,
whereas large values of \(\alpha\) correspond to increasingly regularized estimators.
Throughout this asymptotic analysis, the interpolation parameter \(\alpha\) 
is assumed to remain fixed as \(n\to\infty\).
Both families belong to the general class of M-estimators. Their asymptotic distribution is given by
\[
\sqrt{n}(\hat{\beta}_\alpha - \beta)
\xrightarrow{d}
\mathcal{N}(0,\,
A_\alpha^{-1} B_\alpha A_\alpha^{-1}),
\]
where
\[
A_\alpha = \mathbb{E}\big[\psi_\alpha'(u) X X^\top\big],
\qquad
B_\alpha = \mathbb{E}\big[\psi_\alpha(u)^2 X X^\top\big],
\]
where \(\psi_\alpha\) denotes the derivative, or subgradient, of the corresponding interpolated loss function.
Since the check loss is not differentiable at the origin, \(\psi_\alpha'\) is understood in the usual M-estimation sense as the almost-everywhere derivative, rather than as a pointwise classical derivative.
Since the interpolated objective does not correspond to a likelihood function, we generally have
\[
A_\alpha^{-1} B_\alpha A_\alpha^{-1}
\succ
I(\beta)^{-1},
\]
so that these estimators \emph{do not achieve the Cramér--Rao lower bound in general}.

\subsection{Interpretation of the interpolation}

The Huber estimator occupies a central position in robust statistics because it provides a compromise between the efficiency of least squares estimation and robustness to outliers. Its loss function is quadratic near the origin and linear in the tails, allowing small residuals to be treated in a Gaussian-like manner while limiting the influence of large deviations \citep{Huber1964,Huber1981}.  
The Huber estimator can be seen as a continuous bridge between the sample mean, obtained when the tuning constant \(k \to \infty\), and the sample median, obtained when \(k \to 0\). For finite \(k\), it down‑weights extreme observations, providing robustness while maintaining high efficiency under normality.
In contrast, quantile regression is based on the asymmetric absolute loss associated with the asymmetric Laplace distribution. This estimator is naturally adapted to asymmetric and heavy-tailed models, but its non-differentiability may induce additional finite-sample variability.

The interpolated estimators introduced in this paper can therefore be interpreted as continuous transitions between quantile regression and smoother regularized procedures.
For both the quadratic and Huber interpolations, the parameter \(\alpha\) controls the amount of smoothing introduced into the quantile loss. When \(\alpha=0\), both constructions reduce to the classical quantile estimator. As \(\alpha\) increases, the estimator becomes progressively more regularized.
In the quadratic interpolation, the estimator moves toward a least-squares-type behavior associated with quadratic penalization. In the Huber interpolation, the estimator moves toward the Huber M-estimator, which preserves linear tail behavior and therefore remains compatible with robust Laplace-type geometry.

For small values of \(\alpha\), the estimator remains close to quantile regression and preserves its robustness and asymmetric structure. 
For larger values of \(\alpha\), the estimator behaves increasingly like the smooth, robust Huber estimator. 
This allows the practitioner to fine‑tune the estimator’s balance between robustness against outliers (favoring larger \(\alpha\)) and adherence to a specific quantile structure (favoring smaller \(\alpha\)).

The numerical experiments below illustrate how these two types of interpolation affect finite-sample efficiency. In particular, they show that Huber-type smoothing may improve finite-sample performance while remaining compatible with the geometry of the asymmetric Laplace model, whereas purely quadratic regularization may introduce excessive smoothing in this setting.

\subsection{Finite-sample study across quantile levels}

To complement the asymptotic theory developed in the previous sections, we conducted a comprehensive Monte Carlo study to investigate the finite-sample behavior of the proposed interpolated estimator over the entire range of quantile levels. Rather than focusing on a single quantile, the experiment evaluates the performance of the estimator jointly as a function of both the target quantile level and the interpolation parameter.
We consider the linear regression model
\[
Y_i = 2 + 3X_i + \varepsilon_i,
\]
where
$
X_i \sim \mathcal N(0,1),
$
and where the error term follows an asymmetric Laplace distribution with quantile parameter $\tau$ and unit scale parameter,
$
\varepsilon_i \sim ALD(\tau,1).
$
The true regression parameter is therefore
$
\beta=(2,3)^\top.
$
For each quantile level
$
\tau\in
\{0.05,0.10,\ldots,\\0.95\},
$
we generated $5000$ independent Monte Carlo samples of size
$
n=1000.
$
For every simulated sample, we first computed the standard quantile regression estimator, corresponding to $\alpha=0$, and subsequently used this estimator as the initial value for the numerical optimization of the interpolated estimator.

The interpolation parameter was varied over the grid
$
\alpha\in
\{0,0.05,\ldots,1\},
$
using the Huber interpolation
$
(1-\alpha)\rho_\tau(u)+
\alpha\rho_H(u),
$
where $\rho_H$ denotes the Huber loss with tuning constant
$
k=1.345.
$
This is the classical Huber tuning constant, widely used in robust regression because it provides approximately $95\%$ efficiency under Gaussian errors.
For every pair $(\tau,\alpha)$ we computed the empirical covariance matrix
$
\widehat\Sigma_\alpha(\tau),
$
together with the empirical deviation from the classical quantile regression parameter
$
\widehat b_\alpha(\tau)
=
\overline\beta_\alpha-\beta,
$
where $\overline\beta_\alpha$ denotes the Monte Carlo average of the estimated regression parameter and \(\beta=(2,3)^\top\) is the regression parameter associated with the classical quantile regression model. Consequently, \(\widehat b_\alpha(\tau)\) measures the empirical deviation from the original quantile regression parameter. For \(\alpha>0\), this deviation may reflect not only finite-sample estimation error but also a possible difference between the population minimizer of the interpolated objective and the population parameter targeted by classical quantile regression.
From these quantities we evaluated both the total variance
$
V_\alpha(\tau)
=
\operatorname{tr}
\!\left(
\widehat\Sigma_\alpha(\tau)
\right),
$
and the total mean squared error
$
M_\alpha(\tau)
=
V_\alpha(\tau)
+
\|\widehat b_\alpha(\tau)\|^2.
$
To facilitate comparison with ordinary quantile regression, we report the normalized variance ratio
\[
R_{\mathrm{var}}(\tau,\alpha)
=
\frac{V_\alpha(\tau)}
{V_0(\tau)},
\]
and the normalized mean squared error ratio
\[
R_{\mathrm{MSE}}(\tau,\alpha)
=
\frac{M_\alpha(\tau)}
{M_0(\tau)},
\]
where the denominator always corresponds to the classical quantile regression estimator ($\alpha=0$). Consequently, values below one indicate an improvement over the standard quantile estimator, whereas values above one indicate a deterioration.
Finally, for each quantile level we determined two simulation-based optimal interpolation parameters. Under the variance criterion we define
$
\alpha_{\mathrm{var}}^{*}(\tau)
=
\arg\min_{\alpha}
R_{\mathrm{var}}(\tau,\alpha),
$
whereas under the mean squared error criterion we define
$
\alpha_{\mathrm{MSE}}^{*}(\tau)
=
\arg\min_{\alpha}
R_{\mathrm{MSE}}(\tau,\alpha).
$
These optimal values are obtained retrospectively over the simulation grid and are intended only to summarize the finite-sample behavior of the interpolated estimator. They do not constitute a data-driven rule for selecting the interpolation parameter from a single observed dataset.

\begin{figure}[!htbp]
	\centering
	
	\begin{minipage}{0.49\textwidth}
		\centering
		\includegraphics[width=\linewidth]{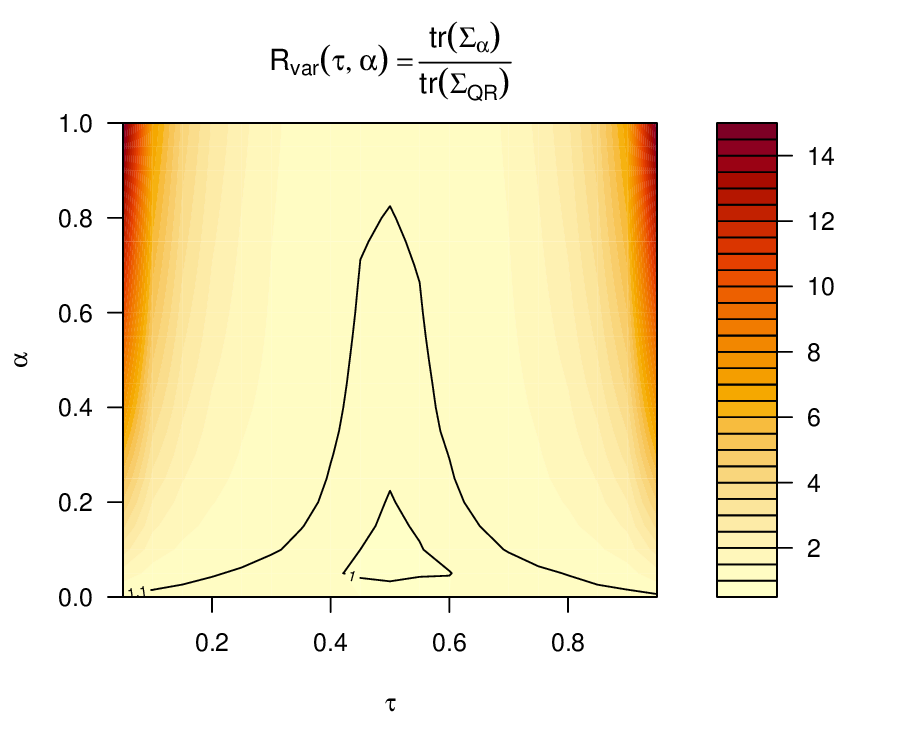}
	\end{minipage}
	\hfill
	\begin{minipage}{0.49\textwidth}
		\centering
		\includegraphics[width=\linewidth]{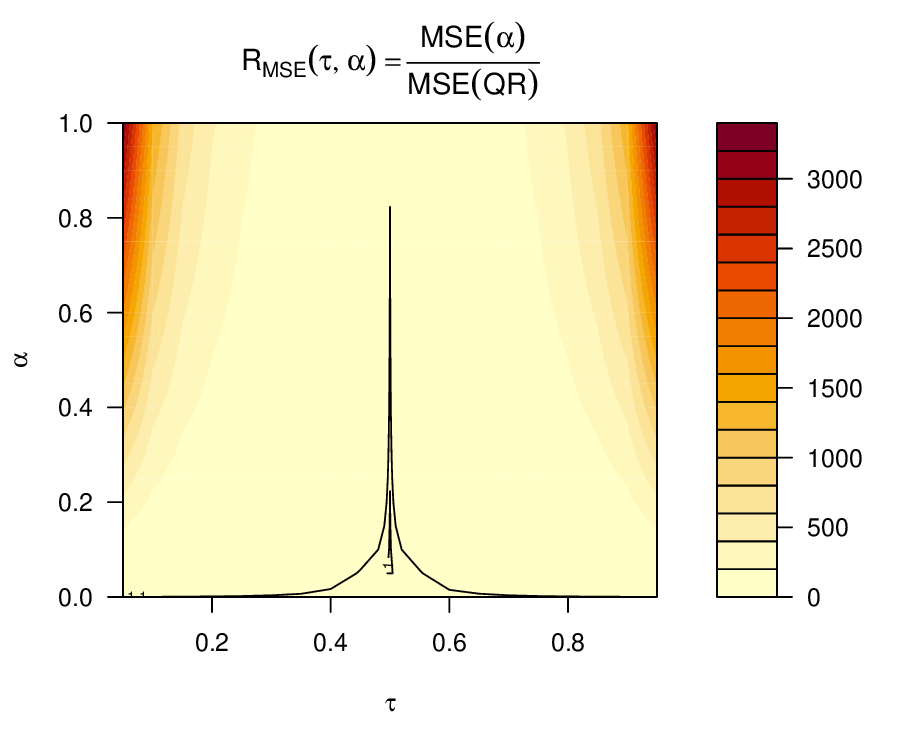}
	\end{minipage}
	
	\caption{
		Contour plots of the normalized variance ratio
		$R_{\mathrm{var}}(\tau,\alpha)$
		(left)
		and normalized mean squared error ratio
		$R_{\mathrm{MSE}}(\tau,\alpha)$
		(right)
		over the $(\tau,\alpha)$ parameter space.
		Values below one correspond to an improvement over the standard quantile regression estimator.
	}
	
	\label{fig:contours}
\end{figure}

\begin{figure}[!htbp]
	\centering
	
	\includegraphics[width=.95\textwidth]{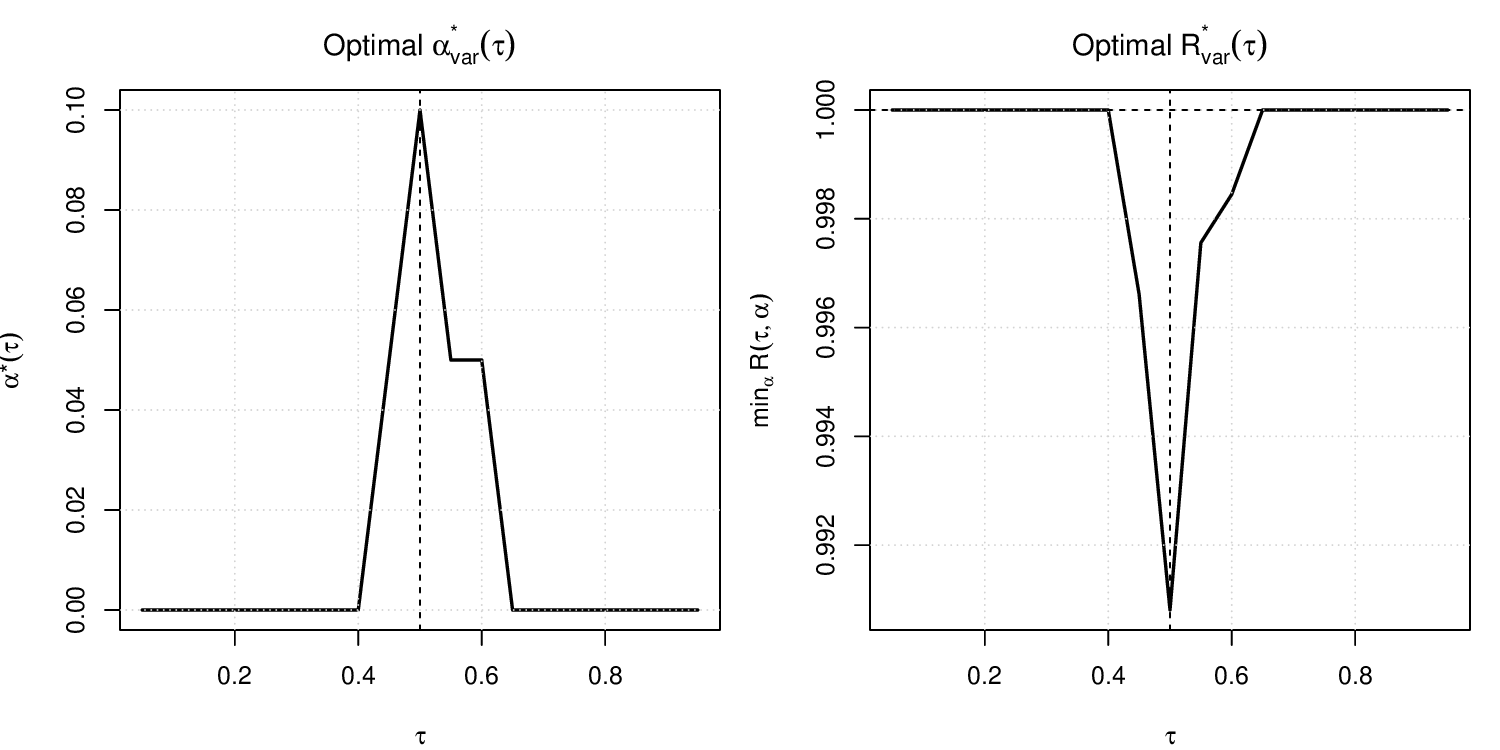}
	
	\vspace{3mm}
	
	\includegraphics[width=.95\textwidth]{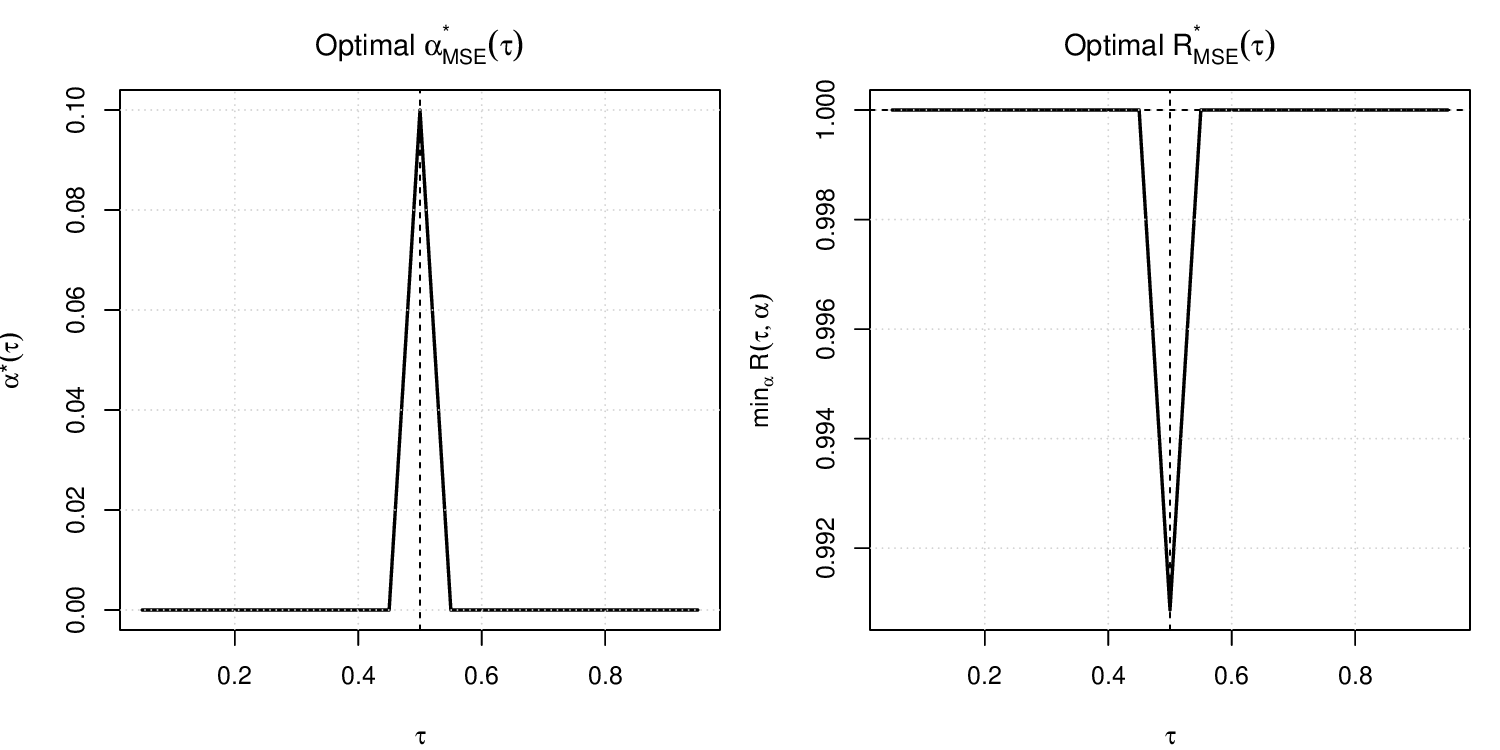}
	
	\caption{
		Simulation-based optimal interpolation parameters \(\alpha_{\mathrm{var}}^{*}(\tau)\) (top) and \(\alpha_{\mathrm{MSE}}^{*}(\tau)\) (bottom), together with the corresponding minimum values of the normalized variance ratio and normalized mean squared error ratio.
	}
	
	\label{fig:optimal}
\end{figure}

Figures~\ref{fig:contours} and~\ref{fig:optimal} summarize the finite-sample behavior of the interpolated estimator throughout the $(\tau,\alpha)$ parameter space.
The contour plot of the variance ratio shows that
$
R_{\mathrm{var}}(\tau,\alpha)
\ge 1
$
over almost the entire parameter space. A small region of improvement appears only in a narrow neighborhood of the median quantile, where moderate interpolation produces a slight reduction of the empirical variance relative to ordinary quantile regression. The minimum variance ratio is attained near
$
\tau=0.5,
$
and
$
\alpha\approx0.10,
$
where the empirical variance is reduced by approximately one percent. Outside this neighborhood, increasing the interpolation parameter leads to a monotone increase of the variance.
The behavior becomes even more pronounced when the mean squared error is considered. The contour plot of
$
R_{\mathrm{MSE}}(\tau,\alpha)
$
shows that values below one occur only in an extremely small neighborhood centered around
$
(\tau,\alpha)\approx(0.5,0.10).
$
Everywhere else, the normalized mean squared error exceeds that of the standard quantile regression estimator. Thus, once the empirical deviation from the classical quantile regression parameter is incorporated into the performance criterion, the finite-sample benefit of interpolation becomes even more localized.

The deterioration becomes increasingly severe as the target quantile moves away from the median. For extreme quantiles such as
$
\tau=0.05
$
or
$
\tau=0.95,
$
large interpolation parameters generate substantial variance inflation, while the corresponding mean squared error increases dramatically owing to the accumulation of both variance and
 deviation from the classical quantile regression parameter. In particular, for $\alpha$ close to one, the normalized variance ratio exceeds an order of magnitude, whereas the normalized mean squared error may become several thousand times larger than that of ordinary quantile regression.

Figure~\ref{fig:optimal} further illustrates this phenomenon through the simulation-based optimal interpolation parameters \(\alpha_{\mathrm{var}}^{*}(\tau)\) and \(\alpha_{\mathrm{MSE}}^{*}(\tau)\).
Under the variance criterion, the simulation-based optimal interpolation parameter \(\alpha_{\mathrm{var}}^{*}(\tau)\) equals zero for almost every quantile level. Only a narrow interval around the median selects a positive interpolation parameter, and even there the optimal value never exceeds
$
\alpha=0.10.
$
The corresponding minimum variance ratio remains very close to one, indicating that the attainable improvement is modest.
The same conclusion is reinforced under the mean squared error criterion. The simulation-based optimal interpolation parameter \(\alpha_{\mathrm{MSE}}^{*}(\tau)\) is identically zero for all quantile levels except the median. At \(\tau=0.5\),
$
\alpha_{\mathrm{MSE}}^{*}(0.5)=0.10,
$
yielding only a slight reduction of the empirical mean squared error.

Overall, the numerical study reveals that the proposed interpolation does not provide a uniform improvement over standard quantile regression. Its finite-sample benefit is confined to a very narrow neighborhood of the median, where a small amount of Huber regularization produces a modest reduction in both variance and mean squared error. As the target quantile moves toward the tails of the distribution, the symmetric Huber regularization becomes increasingly incompatible with the asymmetric geometry of the check loss, leading to a rapid increase in the deviation from the original quantile regression parameter and consequently to a substantial deterioration of finite-sample performance.
These observations complement the asymptotic theory developed in the previous sections. Although the interpolated estimator remains asymptotically consistent and asymptotically normal for every fixed quantile level, the finite-sample behavior depends strongly on the location of the target quantile. The simulation results therefore suggest that Huber interpolation should primarily be viewed as a local regularization strategy for central quantiles rather than as a universally superior alternative to ordinary quantile regression.

\subsection{Discussion and interpretation}

The finite-sample experiments confirm that the effect of interpolation depends
strongly on the target quantile.
Although the interpolated estimator remains asymptotically valid for every
fixed $\tau\in(0,1)$, the simulations show that the finite-sample benefit of
Huber regularization is confined to a very narrow neighborhood of the median.
For almost all quantile levels, the classical quantile regression estimator
already provides the smallest empirical variance and mean squared error.

This behavior can be explained by the interaction between the geometry of the
check loss and that of the Huber loss.
Near the median, where the asymmetric Laplace distribution is nearly symmetric, the Huber approximation introduces only a mild regularization, leading to a slight reduction in finite-sample variability without substantially increasing the deviation from the classical quantile regression parameter. As the target quantile moves away from the median, however, the increasing asymmetry of the check loss becomes progressively less compatible with the symmetric Huber regularization. The observed deterioration appears to be driven not only by increased variability but also by an increasing deviation from the original quantile regression parameter. Indeed, for \(\alpha>0\), the interpolated estimator minimizes a different population objective from that of classical quantile regression, so part of the observed deviation may correspond to a target bias rather than solely to finite-sample estimation error.

From a practical perspective, these results indicate that Huber interpolation
should not be regarded as a universal improvement over ordinary quantile
regression.
Rather, it constitutes a local regularization technique whose potential
benefit is limited to central quantiles and remains quantitatively modest even
there.
Under the correctly specified asymmetric Laplace model, the classical quantile
regression estimator already exhibits nearly optimal finite-sample behavior
over most of the quantile range.

These observations also suggest that future developments should focus on
regularization schemes that preserve the intrinsic asymmetry of the quantile
loss.
In particular, asymmetric smooth interpolations adapted to the target quantile
may retain the numerical advantages of regularization while avoiding the bias
introduced by symmetric smoothing at extreme quantiles.

\section{Numerical Illustrations}
\label{sec:cac40}

To illustrate the generalized quantile interpretation introduced in Subsection~\ref{subsec:inter_estim}, we consider a real-data experiment based on daily log-returns of the CAC-40 index~\citep{quantmod}. Financial return series are known to exhibit asymmetry, heavy tails, and occasional extreme observations, making them particularly well suited for studying robust and quantile-based estimators.
Let
\[
Y_t
=
\log(P_t)-\log(P_{t-1})
\]
denote the daily log-return computed from the CAC-40 closing prices. 
Figure~\ref{fig:CAC40_diagnostics} displays the CAC-40 closing prices, the associated daily log-returns, and the empirical distribution of the returns.
The histogram of returns reveals clear heavy-tailed behavior together with noticeable asymmetry, which motivates the use of robust and quantile-based estimators rather than purely mean-based procedures.
\begin{figure}[!htbp]
	\centering
	\includegraphics[width=\textwidth]{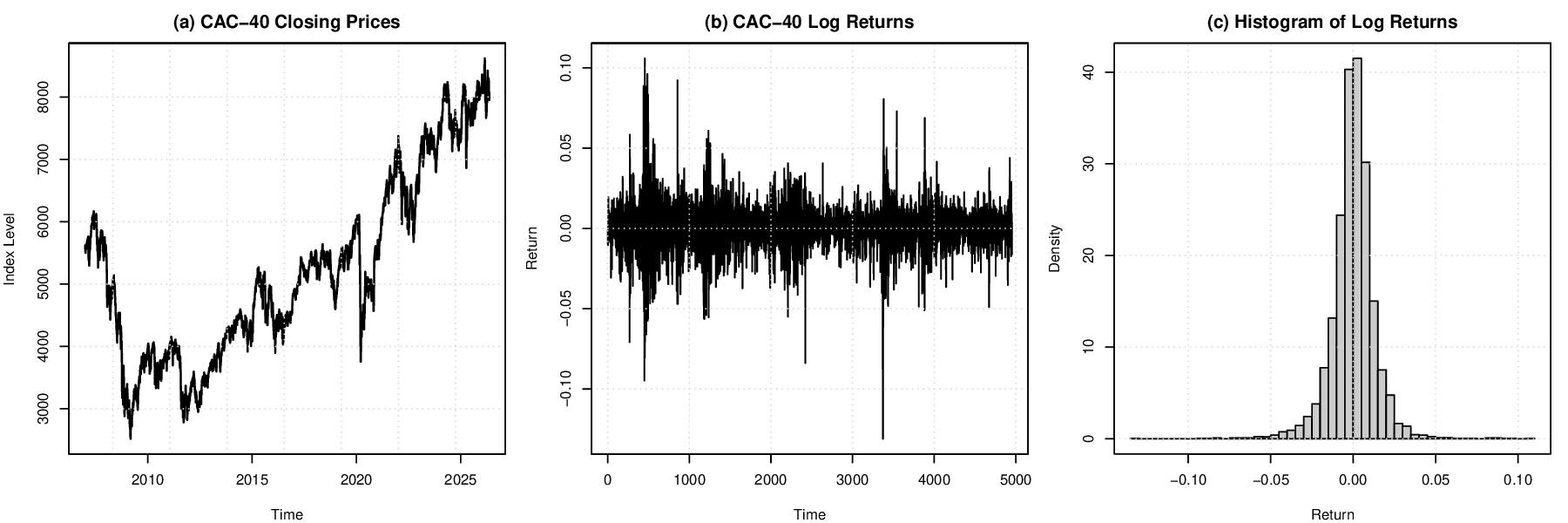}
	\caption{
		Diagnostic plots for the CAC-40 data.
		(a) Daily closing prices of the CAC-40 index.
		(b) Daily log-returns.
		(c) Histogram of daily log-returns.
		The return distribution exhibits heavy tails and mild asymmetry, motivating the use of robust and quantile-based estimators.
	}
	\label{fig:CAC40_diagnostics}
\end{figure}

\noindent
We fix the target quantile level
$
\tau=0.9.
$
The classical empirical quantile estimator therefore corresponds to the empirical quantile of order \(0.9\), associated with the upper \(10\%\) tail of the return distribution.
For the quadratic interpolation, we consider the penalized estimator
\[
\hat q(\tau,h)
=
\arg\min_{q\in\mathbb R}
\left\{
\sum_{i=1}^n
\rho_\tau(y_i-q)
+
\frac h2
\sum_{i=1}^n(y_i-q)^2
\right\}.
\]
The corresponding first-order condition is
\[
F_n(\hat q(\tau,h))
=
\tau
+
h\bigl(\hat q(\tau,h)-\bar Y\bigr).
\]
For each value of the regularization parameter \(h\), we compute the effective empirical quantile order
\[
p(\tau,h)
=
F_n(\hat q(\tau,h))
=
\frac1n
\sum_{i=1}^n
\mathbf 1_{\{y_i\le \hat q(\tau,h)\}},
\]
together with the associated classical empirical quantile
$
y_{(np(\tau,h))}.
$
According to the interpretation developed in Subsection~\ref{subsec:inter_estim}, the quantity \(p(\tau,h)\) represents the empirical probability level naturally associated with the interpolated estimator.
Figure~\ref{fig:generalized_quantile_quadratic} displays the interpolated estimator \(\hat q(\tau,h)\), the effective empirical order \(F_n(\hat q(\tau,h))\), and the difference $\hat q(\tau,h)-y_{(np(\tau,h))}.$
For \(h=0\), the estimator coincides with the classical empirical quantile estimator and therefore satisfies
$
F_n(\hat q_0)
\approx
\tau=0.9.
$
As the regularization parameter increases, the estimator moves continuously toward the sample mean. Consequently, the associated empirical order decreases progressively toward
$
F_n(\bar Y),
$
which corresponds to the empirical order of the sample mean.
Numerically, the experiment shows a smooth transition from
$
F_n(\hat q_0)
\approx
0.90
$
toward
$
F_n(\bar Y)
\approx
0.61.
$
An important observation is that the interpolated estimator \(\hat q(\tau,h)\) and the associated empirical quantile \(y_{(np(\tau,h))}\) are visually almost indistinguishable. The third panel therefore displays their difference
$
\hat q(\tau,h)-y_{(np(\tau,h))},
$
which remains extremely small throughout the interpolation path, typically of magnitude between \(10^{-6}\) and \(10^{-5}\). This confirms that the interpolated estimator behaves numerically as a smooth generalized empirical quantile, while still differing slightly from the discrete empirical quantile induced by the order statistic structure.
We next consider the Huber interpolation defined by
$
(1-\alpha)\rho_\tau(u)
+
\alpha\rho_H(u),
$
or equivalently, for \(\alpha<1\),
\[
\rho_\tau(u)
+
h\rho_H(u),
\qquad
h=\frac{\alpha}{1-\alpha}.
\]
where \(\rho_\tau\) denotes the quantile check loss and \(\rho_H\) denotes the Huber loss.
For each interpolation parameter \(\alpha\in[0,1]\), we compute the corresponding estimator \(\hat q_\alpha\), its effective empirical order
$
p(\alpha)
=
F_n(\hat q_\alpha),
$
and the associated classical empirical quantile
$
y_{(np(\alpha))}.
$
Figure~\ref{fig:generalized_quantile_huber} displays the interpolated estimator \(\hat q_\alpha\), the effective empirical order \(F_n(\hat q_\alpha)\), and the difference $\hat q_\alpha-y_{(np(\alpha))}.$
For \(\alpha=0\), the estimator reduces to the classical empirical quantile estimator and therefore satisfies
$
F_n(\hat q_0)
\approx
0.9.
$
As \(\alpha\) increases, the estimator becomes progressively more influenced by the Huber regularization and moves toward the Huber location estimator. Consequently, the associated empirical order converges toward
$
F_n(\hat\theta_H).
$
Numerically, we obtain
$
F_n(\hat\theta_H)
\approx
0.485,
$
which is close to the empirical median level.
As in the quadratic case, the interpolated estimator and the associated empirical quantile are nearly identical visually, and the third panel reveals only very small deviations between them. The observed differences again remain of very small magnitude, confirming that the Huber interpolation also generates a smooth family of generalized empirical quantiles.
The transition observed in Figure~\ref{fig:generalized_quantile_huber} is substantially more pronounced than in the quadratic case. For moderate values of \(\alpha\), the estimator remains close to the target quantile level \(\tau=0.9\), whereas for values of \(\alpha\) close to \(1\) the empirical order decreases rapidly toward the order associated with the Huber estimator.

Overall, these experiments provide direct empirical support for the interpretation proposed in Subsection~\ref{subsec:inter_estim}. The interpolated estimators introduced in this paper may therefore be viewed as generalized empirical quantiles whose effective order evolves continuously with the regularization parameter. The quadratic interpolation generates a continuous path between the empirical quantile and the sample mean, whereas the Huber interpolation generates a robustified path between the empirical quantile and the Huber M-estimator.

\begin{figure}[!htbp]
	\centering
	\includegraphics[width=\textwidth]{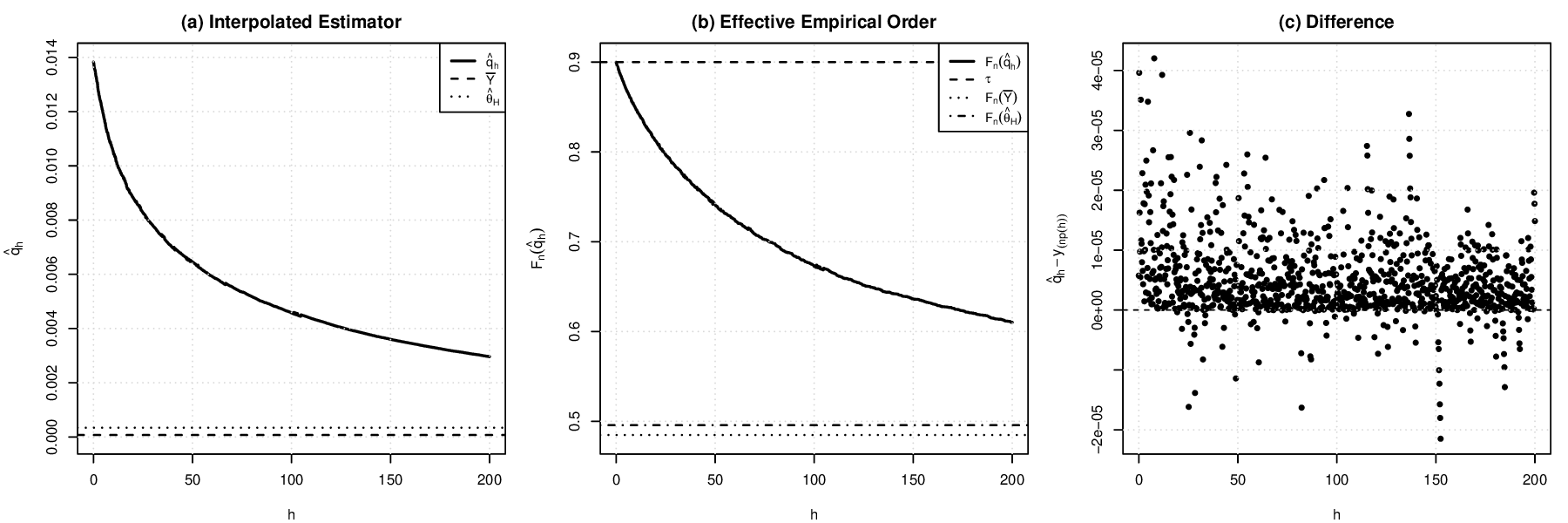}
	\caption{
		Quadratic interpolation for CAC-40 daily log-returns with target quantile level \(\tau=0.9\).
		(a) Evolution of the interpolated estimator \(\hat q(\tau,h)\) as a function of the regularization parameter \(h\).
		The estimator moves continuously from the empirical quantile toward the sample mean.
		(b) Evolution of the effective empirical quantile order \(F_n(\hat q(\tau,h))\).
		The empirical order decreases continuously from the target level \(\tau=0.9\) toward the empirical order of the sample mean.
		(c) Difference between the interpolated estimator and the associated classical empirical quantile,
		$
		\hat q(\tau,h)-y_{(np(\tau,h))}.
		$
		The differences remain extremely small throughout the interpolation path.}
	\label{fig:generalized_quantile_quadratic}
\end{figure}

\begin{figure}[!htbp]
	\centering
	\includegraphics[width=\textwidth]{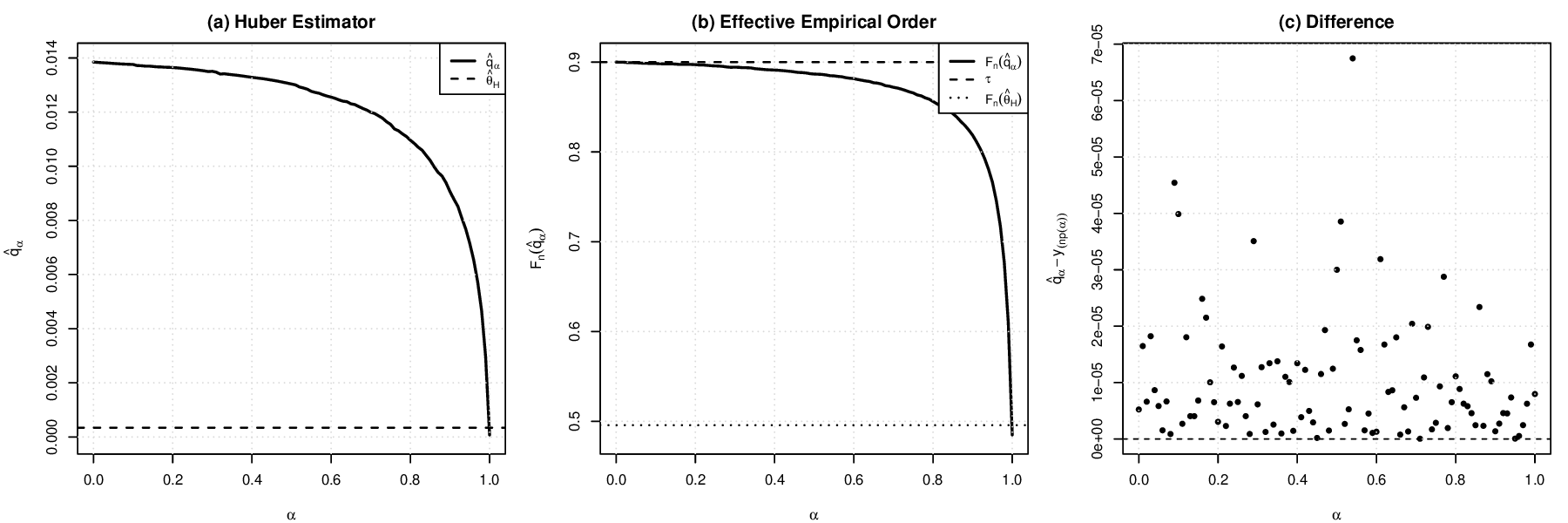}
	\caption{
		Huber interpolation for CAC-40 daily log-returns with target quantile level \(\tau=0.9\).
		(a) Evolution of the Huber-interpolated estimator \(\hat q_\alpha\) as a function of the interpolation parameter \(\alpha\).
		(b) Evolution of the effective empirical quantile order \(F_n(\hat q_\alpha)\).
		The estimator evolves continuously from the empirical quantile toward the Huber location estimator, while the associated empirical order moves from \(\tau=0.9\) toward the empirical order of the Huber estimator.
		(c) Difference between the interpolated estimator and the associated classical empirical quantile,
		$
		\hat q_\alpha-y_{(np(\alpha))}.
		$
		The deviations remain very small along the interpolation path.}
	\label{fig:generalized_quantile_huber}
\end{figure}

\paragraph{Competing interests}
The authors declare that they have no competing interests.

\paragraph{Use of AI tools}
During the preparation of this manuscript, the authors used AI tools to assist with language editing, improving the clarity and presentation of the manuscript. The authors reviewed and edited all AI-generated content, take full responsibility for the final version of the manuscript, and confirm that the scientific content, mathematical results, proofs, simulations, and conclusions are their own.

\bibliographystyle{apalike}
\bibliography{references}

\clearpage
\appendix
\section{Asymptotic efficiency of the quadratic interpolation}
\label{sec:quadratic-efficiency}

Let
\[
q(\tau)=F^{-1}(\tau)
\]
denote the theoretical quantile of order \(\tau\), and let
\[
f_0=f(q(\tau))
\]
be the density evaluated at the quantile.
For the classical empirical quantile estimator \(\hat q(\tau)\), corresponding to \(h=0\), the asymptotic variance is
\[
V(0)
=
\frac{\tau(1-\tau)}{f_0^2}.
\]
Consider now the quadratic interpolated estimator
\[
\hat q(\tau,h)
=
\arg\min_q
\left\{
\sum_{i=1}^n
\rho_\tau(y_i-q)
+
\frac{h}{2}
\sum_{i=1}^n
(y_i-q)^2
\right\}.
\]
At the population level, the minimizer \(q_h\) satisfies
\[
F(q_h)-\tau+h(q_h-\mathbb E Y)=0.
\]
In particular,
\[
q_h=q(\tau)+O(h),
\qquad
(h\to0).
\]
Using standard M-estimation theory, the asymptotic variance of \(\hat q(\tau,h)\) may be written as
\[
V(h)
=
\frac{A(h)}{B(h)^2},
\]
where
\[
A(h)
=
\operatorname{Var}
\Bigl(
\psi_\tau(Y-q_h)
+
h(Y-q_h)
\Bigr),
\]
and
\[
B(h)
=
f(q_h)+h.
\]
For small values of \(h\), replacing \(q_h\) by \(q(\tau)\) yields the second-order approximation
\[
V(h)
\approx
\frac{
	\tau(1-\tau)
	+
	2h\mu_\tau
	+
	h^2\sigma^2
}{
	(f_0+h)^2
},
\]
where
\[
\mu_\tau
=
\mathbb E\!\left[\rho_\tau(Y-q(\tau))\right],
\]
and
\[
\sigma^2
=
\mathbb E\!\left[(Y-q(\tau))^2\right].
\]
Differentiating with respect to \(h\), we obtain
\[
\frac{dV}{dh}
=
\frac{
	2\left[
	(\mu_\tau f_0-\tau(1-\tau))
	+
	h(\sigma^2 f_0-\mu_\tau)
	\right]
}{
	(f_0+h)^3
}.
\]
Consequently, the sign of \(dV(h)/dh\) is determined by the affine function
\[
g(h)
=
(\mu_\tau f_0-\tau(1-\tau))
+
h(\sigma^2 f_0-\mu_\tau).
\]
In particular,
\[
\frac{dV}{dh}\Big|_{h=0}<0
\]
if and only if
\[
\mu_\tau f_0
<
\tau(1-\tau).
\]
Therefore, for sufficiently small \(h>0\),
\[
V(h)<V(0)
\]
whenever the above inequality holds.
This condition is naturally satisfied when the density at the target quantile is sufficiently small. Such situations occur for heavy-tailed distributions and for extreme quantiles located in the tails of the distribution.
Table~\ref{tab:variance_reduction_condition} provides a numerical verification of the condition
\[
\mu_\tau f_0
<
\tau(1-\tau)
\]
for several distributions and quantile levels.
For each model, the table reports the theoretical quantile \(q(\tau)\), the density value
\[
f_0=f(q(\tau)),
\]
the quantity
\[
\mu_\tau
=
\mathbb E[\rho_\tau(Y-q(\tau))],
\]
and the product \(\mu_\tau f_0\).
The numerical results confirm that
\[
\mu_\tau f_0
<
\tau(1-\tau)
\]
for Gaussian, Laplace, and Student distributions, particularly at extreme quantile levels where the density at the target quantile becomes small.
Consequently, for sufficiently small \(h>0\), the quadratic interpolation may reduce the asymptotic variance relative to the classical empirical quantile estimator.
\begin{table}[!htbp]
	\centering
	\begin{tabular}{lcccccc}
		\toprule
		Distribution & $\tau$ & $q(\tau)$ & $f_0$ & $\mu_\tau$ & $\mu_\tau f_0$ & $\tau(1-\tau)$ \\
		\midrule
		\multirow{3}{*}{Gaussian}
		& 0.50 & 0.0000 & 0.3989 & 0.3989 & 0.1591 & 0.2500 \\
		& 0.90 & 1.2816 & 0.1755 & 0.1755 & 0.0308 & 0.0900 \\
		& 0.99 & 2.3263 & 0.0267 & 0.0266 & 0.0007 & 0.0099 \\
		\midrule
		\multirow{3}{*}{Laplace}
		& 0.50 & 0.0000 & 0.5000 & 0.4999 & 0.2499 & 0.2500 \\
		& 0.90 & 1.6094 & 0.1000 & 0.2612 & 0.0261 & 0.0900 \\
		& 0.99 & 3.9120 & 0.0100 & 0.0491 & 0.0005 & 0.0099 \\
		\midrule
		\multirow{3}{*}{Student $t_3$}
		& 0.50 & 0.0000 & 0.3676 & 0.5517 & 0.2028 & 0.2500 \\
		& 0.90 & 1.6377 & 0.1025 & 0.2920 & 0.0299 & 0.0900 \\
		& 0.99 & 4.5407 & 0.0059 & 0.0706 & 0.0004 & 0.0099 \\
		\midrule
		\multirow{3}{*}{Student $t_5$}
		& 0.50 & 0.0000 & 0.3796 & 0.4751 & 0.1804 & 0.2500 \\
		& 0.90 & 1.4759 & 0.1283 & 0.2300 & 0.0295 & 0.0900 \\
		& 0.99 & 3.3649 & 0.0109 & 0.0443 & 0.0005 & 0.0099 \\
		\bottomrule
	\end{tabular}
	\caption{
		Numerical verification of the condition
		\(
		\mu_\tau f_0 < \tau(1-\tau)
		\)
		for several distributions and quantile levels.
		The results show that the sufficient condition for local variance reduction is satisfied for heavy-tailed distributions and for extreme quantiles where the density at the target quantile becomes small.
	}
	\label{tab:variance_reduction_condition}
\end{table}
By contrast, under the asymmetric Laplace distribution associated with classical quantile regression likelihood theory, the standard quantile estimator is already asymptotically efficient. In this case, the quadratic interpolation does not improve the asymptotic variance, which is consistent with the numerical experiments reported below.
Moreover, when
\[
\sigma^2 f_0>\mu_\tau,
\]
the function \(V(h)\) possesses a unique minimum located at
\[
h^\star
=
\frac{
	\tau(1-\tau)-\mu_\tau f_0
}{
	\sigma^2 f_0-\mu_\tau
}.
\]
Thus, the interpolation parameter may be interpreted as a smoothing parameter balancing robustness and asymptotic efficiency.

\section{Proofs of asymptotic normality results}

\subsection{Proof of Theorem \ref{thm:quadratic-clt}}
\label{app:proof-quadratic-clt}

Let $q_0 = q_Q(\tau,h)$ denote the unique solution of the
interpolation equation~\eqref{eq:interpolation-equation}, whose
existence and uniqueness follow from
Proposition~\ref{prop:parametrization}.
Since the indicator function
$
\mathbf 1_{\{y<q\}}
$
is discontinuous in \(q\) at \(y=q\), the score function
\[
\Psi(y,q)
=
\tau-\mathbf 1_{\{y<q\}}+h(y-q)
\]
is not of class \(C^1\) in \(q\).
Consequently, the standard asymptotic normality theorem for smooth
M-estimating equations does not apply directly.
The proof below follows the empirical-process approach for
M-estimators developed in Chapters~5 and~19 of
\citet{Vaart1998}.
Define the empirical and population mean functions
\[
M_n(q) = \frac{1}{n}\sum_{i=1}^n \Psi(Y_i,q),
\qquad
M(q) = \mathbb{E}[\Psi(Y,q)] = \tau - F(q) + h(m-q).
\]
The estimator $\hat{q}_Q$ satisfies $M_n(\hat{q}_Q) = 0$
and $q_0$ satisfies $M(q_0) = 0$.

\paragraph{Step 1: Differentiability of the population criterion.}

The map $q \mapsto M(q) = \tau - F(q) + h(m-q)$ is differentiable
with derivative
\[
M'(q) = -f(q) - h.
\]
This derivative is obtained by differentiating the
\emph{expectation} $M(q)$ with respect to $q$, not by
differentiating the discontinuous integrand $\Psi(y,q)$
pointwise. Specifically, $\tfrac{d}{dq}F(q) = f(q)$ holds at
every point where $F$ admits a continuous density, by the
definition of density; no differentiability of $\psi_\tau$ is
required. Since $f(q_0) > 0$ and $h \ge 0$,
\[
|M'(q_0)| = f(q_0) + h > 0,
\]
so $q_0$ is a locally identifiable zero of $M$.

\paragraph{Step 2: Consistency.}

The class of functions
$\mathcal{F} = \{y \mapsto \mathbf{1}_{\{y<q\}} : q \in \mathbb{R}\}$
is a Glivenko--Cantelli class
\citep[Chapter~19]{Vaart1998},
so $\sup_{q \in \mathbb{R}} |F_n(q) - F(q)| \to 0$
almost surely. The assumption $\mathbb{E}[Y^2] < \infty$ ensures
that the strong law of large numbers applies to the linear term
$h(Y-q)$ uniformly on compact sets. Together these give
\[
\sup_{q \in K} |M_n(q) - M(q)| \xrightarrow{\mathrm{a.s.}} 0
\]
for every compact $K \subset \mathbb{R}$. Since $q_0$ is the
unique zero of $M$ and $|M'(q_0)| > 0$, standard M-estimator
consistency arguments give
\[
\hat{q}_Q \xrightarrow{p} q_0;
\]
see \citet[Theorem~5.7]{Vaart1998}.

\paragraph{Step 3: Stochastic linearization.}

Since $M_n(\hat{q}_Q) = 0$ and $M(q_0) = 0$, write
\begin{equation}
	\label{eq:proof1_decomp}
	0
	= \bigl(M_n(\hat{q}_Q) - M_n(q_0)\bigr)
	+ \bigl(M_n(q_0) - M(q_0)\bigr).
\end{equation}
For any $q$,
\[
M_n(q) - M_n(q_0)
= -\bigl(F_n(q) - F_n(q_0)\bigr) - h(q - q_0).
\]
Decompose the increment of $F_n$ as
\[
F_n(q) - F_n(q_0)
= \bigl(F(q) - F(q_0)\bigr)
+ \bigl\{(F_n(q) - F(q)) - (F_n(q_0) - F(q_0))\bigr\}.
\]
\textit{Deterministic part.}
Since $f$ is continuous at $q_0$,
\[
F(q) - F(q_0) = f(q_0)(q - q_0) + o(|q - q_0|)
\quad (q \to q_0).
\]
\textit{Stochastic part.}
The class $\mathcal{F}$ is also a Donsker class
\citep[Chapter~19]{Vaart1998},
so the empirical process
$\mathbb{G}_n = \sqrt{n}(F_n - F)$
converges weakly in $\ell^\infty(\mathbb{R})$.
In particular, $\mathbb{G}_n$ is asymptotically
equicontinuous at $q_0$: for every $\varepsilon > 0$,
\[
\lim_{\delta \to 0}
\limsup_{n \to \infty}
P\!\Bigl(
\sup_{|q - q_0| \le \delta}
|\mathbb{G}_n(q) - \mathbb{G}_n(q_0)|
> \varepsilon
\Bigr) = 0.
\]
Since $\hat{q}_Q \xrightarrow{p} q_0$, this equicontinuity
implies $\mathbb{G}_n(\hat{q}_Q) - \mathbb{G}_n(q_0) = o_p(1)$,
and therefore
\[
(F_n(\hat{q}_Q) - F(\hat{q}_Q))
- (F_n(q_0) - F(q_0))
= n^{-1/2}\bigl(\mathbb{G}_n(\hat{q}_Q) - \mathbb{G}_n(q_0)\bigr)
= o_p(n^{-1/2}).
\]
Combining the deterministic and stochastic parts and evaluating
at $q = \hat{q}_Q$,
\[
M_n(\hat{q}_Q) - M_n(q_0)
= -(f(q_0) + h)(\hat{q}_Q - q_0)
+ o_p(|\hat{q}_Q - q_0|)
+ o_p(n^{-1/2}).
\]
Substituting into~\eqref{eq:proof1_decomp},
\begin{equation}
	\label{eq:proof1_expansion}
	(f(q_0) + h)(\hat{q}_Q - q_0)
	= (M_n(q_0) - M(q_0))
	+ o_p(|\hat{q}_Q - q_0|)
	+ o_p(n^{-1/2}).
\end{equation}

\paragraph{Step 4: $\sqrt{n}$-consistency and asymptotic
	normality.}

The summands $\Psi(Y_i, q_0)$ are i.i.d.\ with mean
$M(q_0) = 0$ and finite variance
\[
\sigma_\Psi^2
:= \mathbb{E}[\Psi(Y,q_0)^2]
= \mathbb{E}\bigl[
(\psi_\tau(Y-q_0) + h(Y-q_0))^2
\bigr]
< \infty,
\]
where finiteness follows from $|\psi_\tau| \le 1$ and
$\mathbb{E}[Y^2] < \infty$. By the classical central limit
theorem,
\[
\sqrt{n}(M_n(q_0) - M(q_0))
= \frac{1}{\sqrt{n}}\sum_{i=1}^n \Psi(Y_i,q_0)
\xrightarrow{d} \mathcal{N}(0, \sigma_\Psi^2),
\]
so in particular
$\sqrt{n}(M_n(q_0) - M(q_0)) = O_p(1)$.
Multiply~\eqref{eq:proof1_expansion} by $\sqrt{n}$ and let
$Z_n = \sqrt{n}(\hat{q}_Q - q_0)$:
\begin{equation}
	\label{eq:proof1_sqrt_n}
	(f(q_0) + h)\,Z_n
	= \sqrt{n}(M_n(q_0) - M(q_0))
	+ o_p(|Z_n|) + o_p(1).
\end{equation}
Taking absolute values,
$(f(q_0)+h)|Z_n| \le O_p(1) + o_p(|Z_n|) + o_p(1)$,
which gives
$(f(q_0)+h - o_p(1))|Z_n| = O_p(1)$.
Since $f(q_0)+h > 0$, the factor $f(q_0)+h-o_p(1)$ is bounded
away from zero with probability tending to one, and therefore
$Z_n = O_p(1)$, i.e.,
$\sqrt{n}(\hat{q}_Q - q_0) = O_p(1)$.
It follows that
$o_p(|Z_n|) = o_p(O_p(1)) = o_p(1)$,
and~\eqref{eq:proof1_sqrt_n} reduces to
\[
\sqrt{n}(\hat{q}_Q - q_0)
= \frac{\sqrt{n}(M_n(q_0) - M(q_0))}{f(q_0)+h}
+ o_p(1).
\]
By Slutsky's theorem,
\[
\sqrt{n}(\hat{q}_Q - q_0)
\xrightarrow{d}
\mathcal{N}\!\left(
0,\, \frac{\sigma_\Psi^2}{(f(q_0)+h)^2}
\right).
\]
Substituting the expression of $\sigma_\Psi^2$ and
$q_0 = q_Q(\tau,h)$ yields the stated formula
$\sigma_Q^2(\tau,h)$.

\subsection{Proof of Theorem \ref{thm:huber-clt}}
\label{app:proof-huber-clt}

Let
$
q_0=q_H(\tau,h)
$
denote the unique solution of
\[
F(q)-h\,\mathbb E[\psi_H(Y-q)]
=
\tau,
\]
whose existence and uniqueness follow from
Proposition~\ref{prop:huber-existence}.
Since the indicator function
$
\mathbf 1_{\{y<q\}}
$
is discontinuous in \(q\) at \(y=q\), the score function
\[
\Psi_H(y,q)
=
\tau
-
\mathbf 1_{\{y<q\}}
+
h\,\psi_H(y-q)
\]
is not of class \(C^1\) in \(q\).
Consequently, the standard asymptotic normality theorem for smooth
M-estimators does not apply directly.
The proof is based on the general M-estimation and empirical-process
theory of \citet[Chapters~5 and~19]{Vaart1998}.
Define the empirical and population mean functions
\[
M_n(q)
=
\frac1n
\sum_{i=1}^{n}
\Psi_H(Y_i,q),
\qquad
M(q)
=
\mathbb E[\Psi_H(Y,q)].
\]
Then
\[
M(q)
=
\tau
-
F(q)
+
h\,\mathbb E[\psi_H(Y-q)].
\]
The estimator \(\hat q_H(\tau,h)\) satisfies
$
M_n(\hat q_H)=0,
$
whereas \(q_0\) satisfies
$
M(q_0)=0.
$

\paragraph{Step 1: Differentiability of the population criterion.}

Although the score function itself is not differentiable,
the population function \(M(q)\) is differentiable.
Since \(F\) admits a continuous density \(f\),
$
\frac d{dq}\bigl[-F(q)\bigr]
=
-f(q).
$
Furthermore, the derivative of the Huber score is
$
\psi_H'(u)
=
\mathbf 1_{\{|u|\le k\}}
$
for every \(u\neq\pm k\).
Therefore
\[
\frac{\partial}{\partial q}
\psi_H(Y-q)
=
-
\mathbf 1_{\{|Y-q|\le k\}}
\]
for almost every \(Y\).
Since
\[
\left|
\frac{\partial}{\partial q}
\psi_H(Y-q)
\right|
\le 1,
\]
the dominated convergence theorem implies
\[
\frac d{dq}
\mathbb E[\psi_H(Y-q)]
=
-
\mathbb P(|Y-q|\le k).
\]
Consequently,
\[
M'(q)
=
-f(q)
-
h\,\mathbb P(|Y-q|\le k).
\]
Evaluating at \(q=q_0\) gives
$
M'(q_0)
=
-
A_H(\tau,h),
$
where
$
A_H(\tau,h)
=
f(q_0)
+
h\,\mathbb P(|Y-q_0|\le k).
$
Since \(f(q_0)>0\),
$
A_H(\tau,h)>0.
$
Hence \(q_0\) is a locally identifiable zero of \(M\).

\paragraph{Step 2: Consistency.}

The class
\[
\mathcal F_1
=
\{
y\mapsto \mathbf 1_{\{y<q\}}
:
q\in\mathbb R
\}
\]
is a Glivenko--Cantelli class
\citep[Chapter~19]{Vaart1998}.
Moreover, the family
\[
\mathcal F_2
=
\{
y\mapsto \psi_H(y-q)
:
q\in\mathbb R
\}
\]
is uniformly bounded and Lipschitz in the parameter \(q\).
Therefore it is also Glivenko--Cantelli.
Combining these facts yields
\[
\sup_{q\in K}
|M_n(q)-M(q)|
\xrightarrow{\mathrm{a.s.}}
0
\]
for every compact set \(K\subset\mathbb R\).
Since \(q_0\) is the unique zero of \(M\) and
\[
|M'(q_0)|=A_H(\tau,h)>0,
\]
standard M-estimator consistency arguments
\citep[Theorem~5.7]{Vaart1998}
imply
$
\hat q_H
\xrightarrow{p}
q_0.
$

\paragraph{Step 3: Stochastic linearization.}

Since
\[
M_n(\hat q_H)=0,
\qquad
M(q_0)=0,
\]
we write
\begin{equation}
	\label{eq:huber_decomp}
	0
	=
	\bigl(M_n(\hat q_H)-M_n(q_0)\bigr)
	+
	\bigl(M_n(q_0)-M(q_0)\bigr).
\end{equation}
For any \(q\),
\[
M_n(q)-M_n(q_0)
=
-\bigl(F_n(q)-F_n(q_0)\bigr)
+
h\,R_n(q),
\]
where
\[
R_n(q)
=
\frac1n
\sum_{i=1}^{n}
\Bigl(
\psi_H(Y_i-q)
-
\psi_H(Y_i-q_0)
\Bigr).
\]
For the empirical distribution term,
the same argument as in the proof of
Theorem~\ref{thm:quadratic-clt}
gives
\[
F_n(q)-F_n(q_0)
=
f(q_0)(q-q_0)
+
o_p(|q-q_0|)
+
o_p(n^{-1/2}).
\]
For the Huber term,
the Lipschitz continuity of \(\psi_H\) implies
that
\[
\mathbb E[R_n(q)]
=
-
\mathbb P(|Y-q_0|\le k)
(q-q_0)
+
o(|q-q_0|).
\]
The stochastic fluctuation of \(R_n(\hat q_H)\) around its
expectation is \(o_p(n^{-1/2})\) by the Donsker theorem applied
to the class \(\mathcal F_2\), since
\(\hat q_H\xrightarrow{p}q_0\).
Therefore
\[
R_n(q)
=
-
\mathbb P(|Y-q_0|\le k)
(q-q_0)
+
o_p(|q-q_0|)
+
o_p(n^{-1/2}),
\]
uniformly in a neighborhood of \(q_0\).
Combining both expansions and evaluating at
\(q=\hat q_H\) yields
\[
M_n(\hat q_H)-M_n(q_0)
=
-
A_H(\tau,h)
(\hat q_H-q_0)
+
o_p(|\hat q_H-q_0|)
+
o_p(n^{-1/2}).
\]
Substituting into
\eqref{eq:huber_decomp}
gives
\begin{equation}
	\label{eq:huber_expansion}
	A_H(\tau,h)
	(\hat q_H-q_0)
	=
	(M_n(q_0)-M(q_0))
	+
	o_p(|\hat q_H-q_0|)
	+
	o_p(n^{-1/2}).
\end{equation}

\paragraph{Step 4: \(\sqrt n\)-consistency and asymptotic normality.}

The random variables
\[
\Psi_H(Y_i,q_0)
=
\tau
-
\mathbf 1_{\{Y_i<q_0\}}
+
h\,\psi_H(Y_i-q_0)
\]
are i.i.d. with mean zero and variance
$
B_H(\tau,h)
=
\mathbb E
\Bigl[
\Psi_H(Y,q_0)^2
\Bigr].
$
Since \(\psi_H\) is bounded,
$
B_H(\tau,h)<\infty.
$
Therefore, by the classical central limit theorem,
\[
\sqrt n
\bigl(
M_n(q_0)-M(q_0)
\bigr)
=
\frac1{\sqrt n}
\sum_{i=1}^{n}
\Psi_H(Y_i,q_0)
\xrightarrow d
\mathcal N
\!\left(
0,
B_H(\tau,h)
\right).
\]
In particular,
$
\sqrt n
\bigl(
M_n(q_0)-M(q_0)
\bigr)
=
O_p(1).
$
Multiplying
\eqref{eq:huber_expansion}
by \(\sqrt n\) and defining
$
Z_n
=
\sqrt n
(\hat q_H-q_0),
$
we obtain
\[
A_H(\tau,h)\,Z_n
=
\sqrt n
\bigl(
M_n(q_0)-M(q_0)
\bigr)
+
o_p(|Z_n|)
+
o_p(1).
\]
Arguing exactly as in the proof of
Theorem~\ref{thm:quadratic-clt},
one obtains
$
Z_n=O_p(1),
$
and therefore
$
o_p(|Z_n|)
=
o_p(1).
$
Hence
\[
\sqrt n
(\hat q_H-q_0)
=
\frac{
	\sqrt n
	(M_n(q_0)-M(q_0))
}
{A_H(\tau,h)}
+
o_p(1).
\]
Applying Slutsky's theorem yields
\[
\sqrt n
(\hat q_H-q_0)
\xrightarrow d
\mathcal N
\!\left(
0,
\frac{
	B_H(\tau,h)
}
{
	A_H(\tau,h)^2
}
\right).
\]
Finally, recalling that
$
q_0=q_H(\tau,h),
$
we obtain
\[
\sqrt n
\Bigl(
\hat q_H(\tau,h)
-
q_H(\tau,h)
\Bigr)
\xrightarrow d
\mathcal N
\!\left(
0,
\sigma_H^2(\tau,h)
\right),
\]
with
\[
\sigma_H^2(\tau,h)
=
\frac{
	B_H(\tau,h)
}
{
	A_H(\tau,h)^2
}.
\]
This completes the proof.
\qed

\subsection{Proof of Theorem \ref{thm:bisquare-clt}}
\label{app:proof-bisquare-clt}

Let
$
q_0=q_{\mathrm B}(\tau,h)
$
denote the unique solution of
$
\mathbb E[\Psi_h(Y-q)]
=
0,
$
where
\[
\Psi_h(u)
=
\psi_\tau(u)
+
h\psi_B(u).
\]
Since the indicator function appearing in
$
\psi_\tau(u)
=
\tau-\mathbf 1_{\{u<0\}}
$
is discontinuous at \(u=0\), the score function
$\Psi_h(u)$
is not of class \(C^1\).
Consequently, the standard asymptotic normality theorem for smooth
M-estimators does not apply directly.
The proof below follows the same empirical-process strategy as the
proofs of Theorems~\ref{thm:quadratic-clt} and
\ref{thm:huber-clt}.
Define
\[
M_n(q)
=
\frac1n
\sum_{i=1}^{n}
\Psi_h(Y_i-q),
\qquad
M(q)
=
\mathbb E[\Psi_h(Y-q)].
\]
Then
\[
M_n(\hat q_{\mathrm B})=0,
\qquad
M(q_0)=0.
\]

\paragraph{Step 1: Differentiability of the population criterion.}

Although $\Psi_h$ is not differentiable because of the quantile
component, the population function $M$ is differentiable.
Since
$
\mathbb E[\psi_\tau(Y-q)]
=
\tau-F(q),
$
we obtain
\[
\frac d{dq}
\mathbb E[\psi_\tau(Y-q)]
=
-f(q).
\]
Since $\psi_B$ is bounded and continuously differentiable,
dominated convergence yields
\[
\frac d{dq}
\mathbb E[\psi_B(Y-q)]
=
-
\mathbb E[\psi_B'(Y-q)].
\]
Therefore
$
M'(q)
=
-f(q)
-
h\,
\mathbb E[\psi_B'(Y-q)].
$
Evaluating at $q=q_0$ gives
$
M'(q_0)
=
-A_B(\tau,h),
$
where
\[
A_B(\tau,h)
=
f(q_0)
+
h
\mathbb E[\psi_B'(Y-q_0)].
\]
By assumption,
$
A_B(\tau,h)\neq0.
$
Hence $q_0$ is a locally identifiable zero of $M$.

\paragraph{Step 2: Consistency.}

The class
$
\mathcal F_1
=
\{
y\mapsto \mathbf 1_{\{y<q\}}
:
q\in\mathbb R
\}
$
is Glivenko--Cantelli.
Since $\psi_B$ is bounded and Lipschitz,
the class
\[
\mathcal F_2
=
\{
y\mapsto \psi_B(y-q)
:
q\in\mathbb R
\}
\]
is also Glivenko--Cantelli.
Consequently,
\[
\sup_{q\in K}
|M_n(q)-M(q)|
\xrightarrow{\mathrm{a.s.}}
0
\]
for every compact set $K$.
Since $q_0$ is the unique zero of $M$ and
$A_B(\tau,h)\neq0$,
standard M-estimator consistency arguments imply
$
\hat q_{\mathrm B}
\xrightarrow p
q_0.
$

\paragraph{Step 3: Stochastic linearization.}

Since
\[
M_n(\hat q_{\mathrm B})=0,
\qquad
M(q_0)=0,
\]
we write
\[
0
=
(M_n(\hat q_{\mathrm B})-M_n(q_0))
+
(M_n(q_0)-M(q_0)).
\]
For any $q$,
\[
M_n(q)-M_n(q_0)
=
-(F_n(q)-F_n(q_0))
+
hR_n(q),
\]
where
\[
R_n(q)
=
\frac1n
\sum_{i=1}^{n}
\bigl(
\psi_B(Y_i-q)
-
\psi_B(Y_i-q_0)
\bigr).
\]
The empirical-process argument used in the proofs of
Theorems~\ref{thm:quadratic-clt} and
\ref{thm:huber-clt} yields
\[
F_n(q)-F_n(q_0)
=
f(q_0)(q-q_0)
+
o_p(|q-q_0|)
+
o_p(n^{-1/2}).
\]
Since $\psi_B$ is continuously differentiable,
\[
\mathbb E[R_n(q)]
=
-
\mathbb E[\psi_B'(Y-q_0)]
(q-q_0)
+
o(|q-q_0|).
\]
The stochastic fluctuation of \(R_n(\hat q_{\mathrm B})\)
around its expectation is \(o_p(n^{-1/2})\) by the Donsker theorem
applied to the class \(\mathcal F_2\), since
\(\hat q_{\mathrm B}\xrightarrow{p}q_0\).
Therefore
\[
R_n(q)
=
-
\mathbb E[\psi_B'(Y-q_0)]
(q-q_0)
+
o_p(|q-q_0|)
+
o_p(n^{-1/2}),
\]
uniformly in a neighborhood of \(q_0\).
Hence
\[
M_n(\hat q_{\mathrm B})
-
M_n(q_0)
=
-
A_B(\tau,h)
(\hat q_{\mathrm B}-q_0)
+
o_p(|\hat q_{\mathrm B}-q_0|)
+
o_p(n^{-1/2}).
\]
Therefore
\[
A_B(\tau,h)
(\hat q_{\mathrm B}-q_0)
=
(M_n(q_0)-M(q_0))
+
o_p(|\hat q_{\mathrm B}-q_0|)
+
o_p(n^{-1/2}).
\]

\paragraph{Step 4: $\sqrt n$-consistency and asymptotic normality.}

The random variables
\[
\Psi_h(Y_i-q_0)
=
\psi_\tau(Y_i-q_0)
+
h\psi_B(Y_i-q_0)
\]
are i.i.d. with mean zero and variance
$
B_B(\tau,h)
=
\mathbb E
\!\left[
\Psi_h(Y-q_0)^2
\right].
$
Since $\psi_\tau$ and $\psi_B$ are bounded,
$
B_B(\tau,h)<\infty.
$
The classical central limit theorem yields
\[
\sqrt n
(M_n(q_0)-M(q_0))
=
\frac1{\sqrt n}
\sum_{i=1}^{n}
\Psi_h(Y_i-q_0)
\xrightarrow d
\mathcal N
\!\left(
0,
B_B(\tau,h)
\right).
\]
Arguing exactly as in the proofs of
Theorems~\ref{thm:quadratic-clt}
and
\ref{thm:huber-clt},
one obtains
\[
\sqrt n
(\hat q_{\mathrm B}-q_0)
=
\frac{
	\sqrt n(M_n(q_0)-M(q_0))
}
{
	A_B(\tau,h)
}
+
o_p(1).
\]
By Slutsky's theorem,
\[
\sqrt n
(\hat q_{\mathrm B}-q_0)
\xrightarrow d
\mathcal N
\!\left(
0,
\frac{
	B_B(\tau,h)
}{
	A_B(\tau,h)^2
}
\right).
\]
Since
\[
q_0=q_{\mathrm B}(\tau,h),
\]
the stated variance formula follows.
\qed

\end{document}